\documentclass[11pt,letterpaper]{article}

\usepackage[utf8]{inputenc}
\usepackage[T1]{fontenc}
\usepackage{lmodern}
\usepackage[DIV=11]{typearea} 
\usepackage{tikz}

\usepackage{microtype}

\usepackage{amssymb}
\usepackage{amsmath}
\usepackage{amsthm}
\usepackage{thmtools}
\usepackage{thm-restate}

\usepackage{xcolor}
\usepackage{xspace}
\usepackage{xfrac}

\usepackage{comment}
\usepackage[normalem]{ulem}

\usepackage[backend=biber, style=alphabetic, backref=true, doi=false, url=false, maxcitenames=3, mincitenames=3, maxbibnames=10, minbibnames=10, sortlocale=en_US]{biblatex}  

\usepackage[ocgcolorlinks]{hyperref} 
\usepackage{cleveref}
\usepackage[bottom]{footmisc}

\allowdisplaybreaks[1]

\makeatletter
\g@addto@macro\bfseries{\boldmath}
\makeatother

\makeatletter
\g@addto@macro\mdseries{\unboldmath}
\g@addto@macro\normalfont{\unboldmath}
\g@addto@macro\rmfamily{\unboldmath}
\g@addto@macro\upshape{\unboldmath}
\makeatother

\DeclareCiteCommand{\citem}
    {}
    {\mkbibbrackets{\bibhyperref{\usebibmacro{postnote}}}}
    {\multicitedelim}
    {}

\renewcommand*{\multicitedelim}{\addcomma\space}

\AtEveryCitekey{\clearfield{note}}

\newcommand{\myhref}[1]{%
  \iffieldundef{doi}
    {\iffieldundef{url}
       {#1}
       {\href{\strfield{url}}{#1}}}
    {\href{http://dx.doi.org/\strfield{doi}}{#1}}%
}

\DeclareFieldFormat{title}{\myhref{\mkbibemph{#1}}}
\DeclareFieldFormat
  [article,inbook,incollection,inproceedings,patent,thesis,unpublished]
  {title}{\myhref{\mkbibquote{#1\isdot}}}

\makeatletter
\AtBeginDocument{%
    \newlength{\temp@x}%
    \newlength{\temp@y}%
    \newlength{\temp@w}%
    \newlength{\temp@h}%
    \def\my@coords#1#2#3#4{%
      \setlength{\temp@x}{#1}%
      \setlength{\temp@y}{#2}%
      \setlength{\temp@w}{#3}%
      \setlength{\temp@h}{#4}%
      \adjustlengths{}%
      \my@pdfliteral{\strip@pt\temp@x\space\strip@pt\temp@y\space\strip@pt\temp@w\space\strip@pt\temp@h\space re}}%
    \ifpdf
      \typeout{In PDF mode}%
      \def\my@pdfliteral#1{\pdfliteral page{#1}}
      \def\adjustlengths{}%
    \fi
    \ifxetex
      \def\my@pdfliteral #1{}
      \def\adjustlengths{\setlength{\temp@h}{-\temp@h}\addtolength{\temp@y}{1in}\addtolength{\temp@x}{-1in}}%
    \fi%
    \def\Hy@colorlink#1{%
      \begingroup
        \ifHy@ocgcolorlinks
          \def\Hy@ocgcolor{#1}%
          \my@pdfliteral{q}%
          \my@pdfliteral{7 Tr}
        \else
          \HyColor@UseColor#1%
        \fi
    }%
    \def\Hy@endcolorlink{%
      \ifHy@ocgcolorlinks%
        \my@pdfliteral{/OC/OCPrint BDC}%
        \my@coords{0pt}{0pt}{\pdfpagewidth}{\pdfpageheight}%
        \my@pdfliteral{F}
        \my@pdfliteral{EMC/OC/OCView BDC}%
        \begingroup%
          \expandafter\HyColor@UseColor\Hy@ocgcolor%
          \my@coords{0pt}{0pt}{\pdfpagewidth}{\pdfpageheight}%
          \my@pdfliteral{F}
        \endgroup%
        \my@pdfliteral{EMC}%
        \my@pdfliteral{0 Tr}
        \my@pdfliteral{Q}%
      \fi
      \endgroup
    }%
}
\makeatother

\DeclareUnicodeCharacter{221A}{$\sqrt{}$}
\DeclareUnicodeCharacter{3F5}{$\epsilon$}

\colorlet{DarkRed}{red!50!black}
\colorlet{DarkGreen}{green!50!black}
\colorlet{DarkBlue}{blue!50!black}

\hypersetup{
	linkcolor = DarkRed,
	citecolor = DarkGreen,
	urlcolor = DarkBlue,
	bookmarksnumbered = true,
	linktocpage = true
}

\declaretheorem[numberwithin=section]{theorem}
\declaretheorem[numberlike=theorem]{lemma}

\declaretheorem[numberlike=theorem]{corollary}
\declaretheorem[numberlike=theorem]{definition}
\declaretheorem[numberlike=theorem]{claim}

\title{Deterministic Algorithms for Decremental Approximate Shortest Paths: Faster and Simpler}
\author{
Maximilian Probst Gutenberg\thanks{BARC, University of Copenhagen, Universitetsparken 5, Copenhagen 2100, Denmark, The author is supported by Basic Algorithms Research Copenhagen (BARC), supported by Thorup's Investigator Grant from the Villum Foundation under Grant No. 16582.}~\thanks{Work done in part while
visiting the Massachusetts Institute of Technology, Massachusetts, US. The author is supported by STIBOFONDEN’s IT Travel Grant for PhD Students.}
\and
Christian Wulff-Nilsen\thanks{Department of Computer Science, University of Copenhagen. This research is supported by the Starting Grant 7027-00050B from the Independent Research Fund Denmark under the Sapere Aude research career programme.}
}

\date{\today}

\hypersetup{
	pdftitle = {Deterministic Algorithms for Decremental Approximate Shortest Paths: Faster and Simpler},
	pdfauthor = {Maximilian Probst Gutenberg, Christian Wulff-Nilsen}
}

\newcommand{\adist}[2][]{\widetilde{\mathbf{dist}}_{#1}(#2)}
\newcommand{\dist}[2][]{\mathbf{dist}_{#1}(#2)}

\begin{document}

\maketitle

\begin{abstract}
In the decremental $(1+\epsilon)$-approximate Single-Source Shortest Path (SSSP) problem, we are given a graph $G=(V,E)$ with $n = |V|, m = |E|$, undergoing edge deletions, and a distinguished source $s \in V$, and we are asked to process edge deletions efficiently and answer queries for distance estimates $\widetilde{\mathbf{dist}}_G(s,v)$ for each $v \in V$, at any stage, such that $\mathbf{dist}_G(s,v) \leq \widetilde{\mathbf{dist}}_G(s,v) \leq (1+ \epsilon)\mathbf{dist}_G(s,v)$. In the decremental $(1+\epsilon)$-approximate All-Pairs Shortest Path (APSP) problem, we are asked to answer queries for distance estimates $\widetilde{\mathbf{dist}}_G(u,v)$ for every $u,v \in V$. In this article, we consider the problems for undirected, unweighted graphs.

We present a new \emph{deterministic} algorithm for the decremental $(1+\epsilon)$-approximate SSSP problem that takes total update time $O(mn^{0.5 + o(1)})$. Our algorithm improves on the currently best algorithm for dense graphs by Chechik and Bernstein [STOC 2016] with total update time $\tilde{O}(n^2)$ and the best existing algorithm for sparse graphs with running time $\tilde{O}(n^{1.25}\sqrt{m})$ [SODA 2017] whenever $m = O(n^{1.5 - o(1)})$.

In order to obtain this new algorithm, we develop several new techniques including improved decremental cover data structures for graphs, a more efficient notion of the heavy/light decomposition framework introduced by Chechik and Bernstein and the first clustering technique to maintain a dynamic \emph{sparse} emulator in the deterministic setting.

As a by-product, we also obtain a new simple deterministic algorithm for the decremental $(1+\epsilon)$-approximate APSP problem with near-optimal total running time $\tilde{O}(mn /\epsilon)$ matching the time complexity of the sophisticated but rather involved algorithm by Henzinger, Forster and Nanongkai [FOCS 2013]. 
\end{abstract}

\pagebreak

\section{Introduction}

Computing shortest paths in a graph is one of the most classic and well-studied areas in theoretical computer science and algorithms to solve this problem efficiently have found countless applications and are classic subjects of undergraduate algorithm courses, also appearing in CLRS \cite{cormen2009introduction}.

In this article, we focus on maintaining shortest paths in an undirected unweighted \emph{dynamic} graph $G=(V,E)$ that is a graph that is subject to edge insertions/deletions. In particular, we say that a graph undergoing edge insertions is \emph{incremental}, a graph undergoing edge deletions is \emph{decremental} and a graph undergoing edge insertions \emph{and} deletions is \emph{fully-dynamic}. If the graph is either decremental or incremental, it is also referred to as \emph{partially-dynamic}. The goal in this setting is to design a data structure that provides an update operation that takes a single edge as a parameter and then updates the graph and the data structure such that the data structure can be queried for distances in the new version of the graph. 

We focus on two problems: the single-source shortest path (SSSP) problem for \emph{partially-dynamic} graphs where we are given a dedicated source vertex $s \in V$ with the dynamic graph $G=(V,E)$ and where distance queries have to take the form $\mathbf{dist}_G(s,v)$ for any $v \in V$; and the all-pairs shortest path (APSP) problem that we consider only on \emph{decremental} graphs where distance queries can be of any form, i.e. any distance $\mathbf{dist}_G(u,v)$ for any $u,v \in V$ can be queried. In this article, we relax the requirement that distance queries have to be answered exactly and only require that answers to queries are $(1+\epsilon)$-approximate, that is we return some distance estimate $\widetilde{\mathbf{dist}}(u,v)$ that satisfies $\mathbf{dist}(u,v) \leq \widetilde{\mathbf{dist}}(u,v) \leq (1+\epsilon)\mathbf{dist}(u,v)$ where we present algorithms for any $\epsilon > 0$.

Both problems are well-motivated in their setting to model network and navigation problems and appear often as data structures for the more complex fully-dynamic settings or to solve subproblems in static or other dynamic algorithms, for example to compute multi-commodity flows  \cite{madry2010faster}, max-flow and sparsest cuts \cite{Chuzhoy:2019:NAD:3313276.3316320}, to find augmenting paths in matchings \cite{bernstein2018online}, to compute light spanners  \cite{alstrup2017constructing}, or to maintain the diameter and diameter spanners of a partially-dynamic graph \cite{ancona2018algorithms, DBLP:journals/corr/abs-1812-01602}.

We further study the efficient maintenance of dynamic emulators which are central to all current SSSP algorithms for sparse dynamic graphs. In this paper, given a graph $G=(V,E)$, we define a $(d, \epsilon)$-emulator $H$ to be a weighted graph on the vertex set $V$ that has a path from $u$ to $v$ of at most $d$ edges and weight at most $(1+\epsilon)\mathbf{dist}_{G}(u,v)$ for any $u,v \in V$. We further define a $(d, \epsilon)$-hopset $H$ to be a weighted graph on $V$ such that $G \cup H$ is a $(d,\epsilon)$-emulator of $G$. We call hopsets and emulators dynamic, if they are themselves dynamic graphs that satisfy the above-stated property for a dynamic graph $G$ at any stage. 

To ease stating results, we henceforth let $n$ denote the number of vertices in a dynamic graph and $m$ refer to the total number of edges that are in any version of the graph $G$. For algorithms working on weighted graphs, we further denote by $W$ the aspect ratio of the graph (i.e. the largest edge weight divided by the smallest). Whilst results in fully-dynamic graphs are normally stated by the running time of a single update, incremental and decremental algorithms are normally judged by their total running time that is the sum of all update times. We point out that all existing algorithms have logarithmic or constant query time.

\subsection{Related work}

\paragraph{Shortest paths in static graphs.} For static graphs, the exact SSSP problem for positive edge weights can be solved by Dijkstra using Fibonacci heaps in $O(m + n \log n)$ time  \cite{cormen2009introduction} or even in $O(m)$ time using Thorup's algorithm  \cite{thorup1999undirected}. The APSP problem can be solved by running Thorup's algorithm from each source in $O(mn)$ time (for slightly faster algorithms consult \cite{williams2014faster, chan2016deterministic}). The $O(mn)$ time bound is widely believed to be the optimal combinatorial running time up to subpolynomial factors\footnote{Various faster algorithm are known using Matrix Multiplication (for example \cite{zwick2002all, le2012faster}), however these algorithms have not proven to be practically efficient.}. In fact, even improving the running time for approximate APSP for multiplicative stretch $(1 + \epsilon)$ and additive stretch $(2-\epsilon)$ for any $\epsilon \in [0,1]$ by a truly polynomial factor implies major breakthroughs for several long-standing problems \cite{dor2000all, williams2018subcubic, lincoln2018tight}.

\paragraph{Decremental SSSP.} For decremental graphs, we can simply rerun each of these algorithms after every edge deletion incurring an additional multiplicative factor of $O(m)$ in the total running time. The first algorithm to improve on this trivial bound was presented for the exact SSSP problem for unweighted graphs by Even and Shiloach \cite{shiloach1981line} that takes $O(mn)$ total update time and which was extended to work on directed graphs by Henzinger and King \cite{henzinger1995fully}. This data structure often referred to as ES-tree has ever since become a fundamental tool in dynamic graph algorithms. However, recent hardness results \cite{roditty2004dynamic, abboud2014popular, henzinger2015unifying} imply that improving this fundamental barrier would imply a breakthrough on various long-standing problems.

Due to these hardness results, the research community has turned to $(1+\epsilon)$-approximate algorithms for the decremental SSSP problem. After improving the total update time for decremental undirected graphs to $O(n^{2+o(1)})$ \cite{bernstein2011improved}, Henzinger, Forster and Nanongkai developed a near-linear time algorithm taking time $O(m^{1+o(1)} \log W)$ \cite{henzinger2014decremental} for decremental \emph{weighted} graphs. Unfortunately, all these algorithms rely on heavy randomization and the assumption of an \textit{oblivious} adversary which implies that adversarial edge updates to the graph have to be fixed before the algorithm is started. This assumption limits severely the area of applications as it implies that the data structures cannot be used as a black box. 

To overcome this problem, Bernstein and Chechik \cite{chechik2016decremental} presented the first deterministic algorithm running in $\tilde{O}(n^2)$ total time. They also presented an algorithm for sparse graphs \cite{bernstein2017deterministic} with update time $\tilde{O}(n^{1.25}\sqrt{m})=\tilde{O}(mn^{3/4})$. Bernstein \cite{bernstein2017deterministicweighted} later extended the algorithm for dense graphs to handle edge weights in $[1,W]$ with total update time $\tilde{O}(n^2 \log W)$. Each of these algorithms can answer queries for distance estimates in constant time but their fundamental technique of contracting vertices in dense components of the graphs makes it impossible to retrieve a corresponding path. Recently, Chuzhoy and Khanna showed in \cite{Chuzhoy:2019:NAD:3313276.3316320} that the algorithm for weighted dense graphs by Bernstein can be extended to obtain an algorithm with $O(n^{2+o(1)} \log W)$ total update time where approximate shortest-paths can be retrieved in $O(n \log W)$ time. However, their algorithm is randomized, working against an adaptive adversary, and only works in the setting of vertex deletions. 

For the directed setting, Henzinger, Forster and Nanongkai presented an algorithm with running time $O(mn^{0.9 + o(1)}\log W)$ \cite{henzinger2014sublinear, henzinger2015improved} that works against an oblivious algorithm. Recently, Probst Gutenberg and Wulff-Nilsen \cite{probstWulffNilsenDiSSSP} achieved improved running time $\tilde O(mn^{3/4}\log W)$ (with better bounds for a wide range of graph densities) and showed that there exists an algorithm with update time $O(m^{3/4}n^{5/4}\log W)$ that works against an adaptive adversary, however, their latter algorithm cannot report paths. 

\paragraph{Decremental APSP problem.} For the decremental APSP problem, a plethora of algorithms is known \cite{king1999fully, baswana2002improved, demetrescu2004new, roditty2004dynamic,thorup2005worst, bernstein2011improved,  roditty2012dynamic, abraham2013dynamic, henzinger2014decremental, bernstein2016maintaining, henzinger2016dynamic, abraham2017fully, chechik2018near, brand2019dynamic, probstWulffNilsenwcAPSP}. We want to point out in particular the exact deterministic decremental APSP algorithm for weighted digraphs by Demetrescu and Italiano \cite{demetrescu2004new} with running time $\tilde{O}(n^3)$ and the deterministic $(1+\epsilon)$-approximate decremental APSP algorithm by Henzinger, Forster and Nanongkai \cite{henzinger2016dynamic} with total update time $\tilde{O}(mn)$ on undirected, unweighted graphs. These two algorithms dominate all known deterministic approaches for the decremental APSP problem in running time and generality.

\paragraph{Hopsets and Emulators.} Cohen \cite{cohen2000polylog} introduced the notion of hopsets and presented an efficient algorithm to compute a sparse $(O(\text{polylog}(n)), \epsilon)$-hopset $H$ of size $\tilde{O}(n)$ (with a small additive error). After the initiation of the field, Bernstein \cite{bernstein2009fully} observed first that the emulators presented by Thorup and Zwick \cite{thorup2006spanners} are in fact $(n^{o(1)}, \epsilon)$-emulators that can be efficiently maintained and successively obtained an efficient dynamic APSP algorithm upon the emulator instead of the input graph. Since this breakthrough, hopsets have become a fundamental technique for decremental SSSP and have been used in all algorithms for decremental $(1+\epsilon)$-approximate SSSP designed for sparse graphs. Recently, Elkin and Neiman \cite{elkin2016hopsets} presented $(O(1), \epsilon)$-hopsets inspired by techniques for $(1+\epsilon, \beta)$-spanners by Elkin and Peleg \cite{elkin20041} and Pettie and Huang \cite{huang2018thorup} proved that the Thorup-Zwick emulators are universally-optimal hopsets that match the bounds given by Elkin and Neiman. 

\subsection{Our Contributions}

Our main result is a new \emph{deterministic} algorithm for the $(1+\epsilon)$-approximate SSSP problem that runs in $O(mn^{0.5 + o(1)})$ total update time. 

\begin{theorem}
\label{thm:mainSSSP}
There exists a deterministic data structure, for any $\epsilon > 1/\text{polylog}(n)$, that given an unweighted undirected partially-dynamic graph $G=(V,E)$ and a dedicated source, can process edge deletions in total update time $O(mn^{0.5+ O(1/\sqrt{\log n})})$ and can return a distance estimate $\widetilde{\mathbf{dist}}(s,v)$ such that \[
{\mathbf{dist}}(s,v) \leq \widetilde{\mathbf{dist}}(s,v) \leq (1+\epsilon){\mathbf{dist}}(s,v)
\]
for any $v \in V$ in constant worst case time.
\end{theorem}

This, improves the algorithms by Bernstein and Chechik \cite{bernstein2016deterministic, bernstein2017deterministic, bernstein2017deterministicweighted} over the entire sparsity $m = O(n^{1.5 - o(1)})$ and dominates it heavily for the important case when $m = O(n)$ where our algorithm improves the current running time of $\tilde{O}(n^{1+3/4})$ to just $O(n^{1.5 + o(1)})$ total update time. This matches a natural barrier to the problem encountered by all current approaches as pointed out by Bernstein and Chechik \cite{bernstein2017deterministic}. We present the algorithm for the more challenging decremental setting but point out that an adaption to the incremental setting is straight-forward. Unfortunately, as all previous deterministic algorithms, we cannot return approximate shortest-paths.

Further, we present a new algorithm for $(1+\epsilon)$-approximate APSP.

\begin{theorem}
\label{thm:mainAPSP}
There exists a deterministic data structure, for any $\epsilon > 0$, that given an unweighted undirected decremental graph $G=(V,E)$ and a dedicated source, can process edge deletions in total update time $O(mn\log n/\epsilon)$ and can return a distance estimate $\widetilde{\mathbf{dist}}(u,v)$ such that \[
{\mathbf{dist}}(u,v) \leq \widetilde{\mathbf{dist}}(u,v) \leq (1+\epsilon){\mathbf{dist}}(u,v)
\]
for any $u,v \in V$ in $O(\log \log n)$ worst case time.
\end{theorem}

This algorithm matches the best known running time by Henzinger, Forster and Nanongkai  \cite{henzinger2016dynamic} but significantly simplifies upon their data structure and requires no sophisticated proof techniques to bound the running time. 

To achieve these results we design the first non-trivial \emph{sparse} dynamic $(d, \epsilon)$-emulator beyond the approach by Thorup and Zwick  \cite{thorup2006spanners} which is normally maintained using heavy randomization. In order to construct an efficient $(d, \epsilon)$-emulator, we refine techniques to maintain decremental cover structures, extend the heavy/light decomposition as proposed by Bernstein and Chechik \cite{bernstein2017deterministic} and introduce novel modification to ES-trees which were not considered so far and which allow us to exploit the full strength our emulator construction whilst keeping the analysis simple. We sketch and discuss these novel techniques in more detail in the next section.

\subsection{Overview and Techniques}
\label{sec:overview}

We now introduce the most important techniques used in the paper and give a high-level overview of our algorithms\footnote{In this overview, we assume basic familiarity with ES-trees as introduced in  \cite{shiloach1981line} and monotone-ES-trees (MES-tree) as for example described in  \cite{henzinger2016dynamic}. We provide a short introduction in section \ref{subsec:basicDefs} for readers unfamiliar with these data structures.}.

\paragraph{Existing Approaches for Decremental SSSP.} So far, all efficient algorithms  \cite{henzinger2014decremental, henzinger2014sublinear, bernstein2017deterministic, chechik2018near} for \emph{sparse} graphs take roughly the following approach: a dynamic $(d, \epsilon)$-emulator $H$ is constructed and maintained for $G$. Then, an MES-tree is maintained from the source vertex $s$ on $H$ instead of $G$. A simple edge rounding scheme can then reduce the running time to $O(|H|d)$ since we only need to maintain the shortest path tree of $s$ up to depth $d$\footnote{Technically some papers construct $(d, \epsilon)$-hopsets but the union of hopset and input graph $G$ can easily be seen to form an emulator.}. Thus, efficient SSSP algorithms try to find an optimal trade-off between the sparsity of $H$ (i.e. $|H|$), the properties of $H$ (i.e. aiming for a small $d$), and the time to maintain $H$ dynamically. Clearly, our description oversimplifies the rather involved algorithms but provides us with a road-map since we take the same fundamental approach. 

\paragraph{A New Approach for Efficient Emulators.} Whilst constructing dynamic emulators has obtained extensive attention, so far, only an approach based on the Thorup-Zwick emulators \cite{thorup2006spanners} has been proven to provide dynamic sparse emulators (i.e. an emulator $H$ with $|H|= O(m^{1+o(1)})$). But these techniques rely heavily on randomization and therefore do not extend to the deterministic setting. 

Here, we present the first \emph{deterministic} hopset construction and maintenance that ensures that the emulator $H$ is sparse using a density-sensitive clustering approach. Our approach is based on finding a cluster center $c$  for each vertex $v$ at some distance $\tau$ and then by adding edges from $c$ to all vertices at distance $2\tau/\epsilon + \tau$ to $H$. Thus, using only two edges in $H$, $v$ can travel to any vertex at distance $2\tau/\epsilon$ and the weight of this $2$-hop path is at most $2\tau + 2\tau/\epsilon =  2\tau/\epsilon (\epsilon + 1)$ which constitutes only a $(1+\epsilon)$ multiplicative error on the original distance.

This approach of assigning cluster centers was formerly used to construct $k$-spanners \cite{awerbuch1985complexity, peleg1989graph}, $(1+\epsilon, \beta)$-spanners \cite{elkin20041} and  universally optimal hopsets \cite{elkin2016hopsets}, however in the dynamic setting only a single recent approach has tried to adapt these techniques to the dynamic setting which again relied on mixing it with the Thorup-Zwick emulator framework \cite{chechik2018near}.

We overcome the use of randomization by carefully identifying the cluster centers using $\tau$-covers which are introduced in the next paragraph. However, we require tremendous flexibility of $\tau$-covers which no current approach provides. Therefore much of this overview shows how to obtain more robust $\tau$-covers that we can then adapt to our constraints. In particular, we show how to find cluster centers only in sparse subgraphs which can be done more efficiently than in dense graphs. To deal with dense subgraphs we then employ another variant of $\tau$-covers to contract them without inducing large additive error for shortest paths building upon the techniques of Bernstein and Chechik \cite{bernstein2017deterministic}.

In the final two paragraphs, we sketch how to use our emulator $H$ which, whilst following the fundamental approach described above, requires us to adapt MES-trees to support two new techniques which might be of independent interest.

\paragraph{Maintaining Covers on Connected Graphs.} As aforementioned, at the heart of our algorithms is a new technique to maintain covers. Formally, we say that a $\tau$-cover $C$ of a decremental graph $G=(V,E)$ is a subset of the vertices such that at any stage, for any vertex $v \in V$, there is always a vertex $w \in C$ at distance at most $\tau$ from $v$. Naturally, the $\tau$-cover $C$ forms a good set of potential cluster centers: adding edges from each vertex $c \in C$ to all vertices in $B(c, 2\tau/\epsilon + \tau)$ to $H$ guarantees that every vertex has a $2$-hop as described above. Thus, $\tau$-covers arise naturally in the context of our emulator.

To gain some intuition, let us assume that $G$ remains connected for the rest of this paragraph. We recall that ES-trees can be maintained in time $O(m\delta)$ to maintain the shortest path tree $T$ from some source $s$ to depth $\delta$ in $G$. Now, consider we want to find a $\tau$-cover $C$ for $G$. Then, we could run an ES-tree to depth $\tau$ from each vertex $v \in V$, that is we explicitly maintain its ball $B(v,\tau) = \{ u \in V | \mathbf{dist}_G(u,v) \leq \tau\}$. If some vertex $v$ has no vertex of the set $C$ in its ball $B(v,\tau)$, we add $v$ to $C$. It is straight-forward to show that $C$ can be maintained correctly in $O(n m \tau)$ total time. Further, for any two vertices $c, c' \in C$, we have $B(c, \tau/2) \cap B(c', \tau/2) = \emptyset$ by the criterion of adding a vertex to $C$ only when there is no vertex at distance $\tau$. But since these balls are all-pairwise disjoint, and by the \emph{connectedness} of the graph each ball contains at least $\tau/2$ vertices, we have by the pigeonhole principle at most $2n/\tau$ vertices in $C$. Observe that this approach enforces that $C$ is an incremental set, i.e. elements added to $C$ remain in $C$.

It is also straight-forward to layer this approach. Let $C_{0} = V$, and maintain $C_i$ for $i \in [0, \lg n]$ as follows: maintain an ES-tree to depth $2^{i-1}$ from every vertex in $C_{i-1}$. Again, if some vertex $c \in C_{i-1}$ does not have a vertex in $C_i$ in its ES-tree, add it to $C_{i}$. It is easy to prove that each $C_i$ is a $2^i$-cover of size at most $n/2^{i-3}$. Thus, the running time to maintain all $C_i$'s is $O(\sum_{i = 0}^{\lg n} O(|C_{i-1}| m 2^{i-1}) = O(mn \lg n)$. 

\paragraph{Maintaining Covers as Graphs Decompose.} Observe that our analysis crucially relied on the fact that $G$ remained connected. This is because if we allow $G$ to decompose into several connected components, vertices from which we maintain ES-trees could be isolated and this could force us to add more vertices to $C_i$ (and subsequently maintain more ES-trees). 

However, an efficient algorithm to maintain a $\tau$-cover $C$ was recently proposed by Henzinger, Forster and Nanongkai \cite{henzinger2016dynamic} using a sophisticated but complicated scheme that "moves" ES-trees to nearby vertices, if an edge update isolates a vertex that is in $C$. However, this means that $C$ becomes a \emph{dynamic} set, rather than an \emph{incremental} one. 

We show that a very simple trick can be used to maintain $\tau$-covers without "moving" ES-trees: \emph{we ignore adversarial updates that disconnect the graph}. Therefore, we define the decremental graph $G^*$ to be the decremental graph $G$ where we ignore edge deletions that disconnect the graph. Then, we claim that $G^*$ still has the desirable properties that we need and that in turn the information obtained by maintaining the $\tau$-cover on $G^*$ is useful to obtain efficient algorithms. 

For example, consider the $(1+\epsilon)$-approximate APSP problem: we can maintain $2^i$-covers as described above for each $i \in [0, \lg n]$ and maintain from each vertex in $C_i$ an ES-tree to depth ${2^i}/{\epsilon}$. Upon a query for an estimate $\widetilde{\mathbf{dist}}(u,v)$ for any $u,v \in V$, $u$ can check for each $i$, if the vertex $c_i$ closest to $u$ in $C_i$ has $v$ in its ES-tree. Taking the smallest index $i$, it can be shown that $\mathbf{dist}_{G^*}(u, c_i) + \mathbf{dist}_{G^*}(c_i, v)$ has weight at most $(1+4\epsilon)\mathbf{dist}_G(u,v)$. Let $P_{u,v}$ be the corresponding approximate shortest path from $u$ to $v$ via $c_i$. Consider first that $u,v$ and $c_i$ are all in the same connected component, then $P_{u,v}$ is contained in $G$. Otherwise, if $u$ and $v$ are not in the same component their distance is $\infty$ which we can reconstruct using a connectivity data structure. If $c_i$ is in another component than $u$ and $v$, then the path segment from $u$ to $c_i$ in $P_{u,v}$ contains some edge $(x,y)$ leaving the component of $u$, i.e. $(x,y)$ is not in $G$. But on closer inspection, also the path segment from $v$ to $c_i$ must contain the edge $(x,y)$ (otherwise $(x,y)$ would be part of a cycle which by definition implies it would be deleted from $G^*$). But then, $P_{u,v}$ contains a cycle and removing the cycle, leaves us with a simple path from $u$ to $v$ in $G$ that is of weight at most $(1+4\epsilon)\mathbf{dist}_G(u,v)$.

In fact, the described approach allows for a very concise proof for a decremental $(1+\epsilon)$-approximate algorithm for APSP that matches the best known running time: $O(mn \log n/\epsilon)$. 

\paragraph{Faster Covers by Ignoring Dense Subgraphs.} We observe that the $\tau$-cover structure presented so far still requires running time $\tilde{O}(mn)$. But for the decremental approximate SSSP problem, we aim for a significantly better total running time. To improve the running time of the $\tau$-cover structure, we combine our $\tau$-covers with a simple approach suggested by Bernstein and Chechik \cite{bernstein2017deterministic}: We \emph{ignore} dense subgraphs. To formalize dense subgraphs, we define $\mathbf{deg}(B_{G^*}(v,r))$ to denote the number of edges in $G^*$ with both endpoints in the ball $B_{G^*}(v,r)$ on the graph $G^*$. Then, let us say a vertex $v$ is $(\mu, r)$-\textit{heavy} if $\mathbf{deg}(B_{G^*}(v,r)) > \mu$, and otherwise we say $v$ is $(\mu, r)$-\textit{light}. 

We now modify the $\tau$-cover algorithm to run each ES-tree rooted at some $v$ to depth $d$ maintained in the $\tau$-cover algorithm, only if $v$ is $(\mu, d)$-\textit{light}. Consider the effect on a vertex $v$: if $v$ is $(\mu, 2\tau)$-\textit{light}, then there is a vertex $w$ at distance $\tau$ in the $\tau$-cover and that vertex is itself $(\mu, \tau)$-\textit{light} since $B(w, \tau) \subseteq B(v, 2\tau)$ (Recall also that in the previous layer cover there is a vertex $w'$ at distance at most $\tau/2$ from $v$ that joins the $\tau$-cover if there is no vertex at distance at most $\tau$ to $v$). Otherwise, there might not be a vertex close to $v$ in the cover.

This reduces the running time for each ES-tree to $O(\mu d)$ instead of $O(m d)$ and therefore the total running time to maintain the cover reduces to $\tilde{O}(n \mu)$ (see section \ref{subsec:basicDefs} for details). We henceforth denote by $C_{\mu, \tau}$ a particular instance of the set obtained by running this modified $\tau$-cover algorithm (again we only consider $\tau$'s that are powers of $2$, so there are only $O(\log n)$ of these instances). 

\paragraph{An Efficient Threshold Emulator for Sparse Subgraphs.} 
Let us now describe how to construct a dynamic weighted emulator $H_{\mu, \tau, \epsilon}$ for a certain distance threshold $\tau$ (and value $\mu$ to be fixed later such that $m/\mu \approx \sqrt{n}$), that is a dynamic graph with $(m/\mu, \epsilon)$-emulator guarantees for any vertex $u$ at a distance from $s$ in $[m\tau/\mu, 2m\tau/\mu]$. Initially, we assume that every vertex in $V$ is $(\mu, 4\tau (1/\epsilon)^{\sqrt{\log n}+1})$-\textit{light} (abbreviated simply by \textit{light} for the rest of the overview and analogously we write \textit{heavy} for vertices that are not \textit{light}).

Therefore, we maintain a hierarchy with levels $0, 1, .. , k = \lfloor \sqrt{\log n} \rfloor$ levels.  We assign each center $c \in C_{\mu, \tau}$ a level in the hierarchy denoted by $\textsc{ClusterLevel}_{\mu, \tau}(c)$. We later show that we can maintain $i = \textsc{ClusterLevel}_{\mu, \tau}(c)$ such that it is monotonically decreasing over time and satisfies that $|B_{G^*}(c, \tau(1/\epsilon)^{i+1})| \leq n^{1/k} |B_{G^*}(c, \tau(1/\epsilon)^{i})|$.

For each level $i \in [0,k]$, we then maintain a set of \textit{active} centers $A_i$ that is a maximal subset of the vertices in $C_{\mu, \tau}$ at cluster level $i$ such that no two vertices in $A_i$ are at distance less than $2\tau(1/\epsilon)^{i}$. In particular, the balls $B_{G^*}(c, \tau(1/\epsilon)^{i})$ for all $c \in A_i$ are pairwise disjoint.

Finally, we ensure that $H^s_{\mu, \tau, \epsilon}$ contains for every active center $c \in A_i$ the edges $\{c\} \times B_{G^*}(c, \tau(1/\epsilon)^{i+1})$ and the edges in $\{s\} \times B_{G^*}(s, 2\tau(1/\epsilon)^{k+1})$ where each edge is assigned weight $w(u,v) = \lceil \mathbf{dist}_{G^*}(u,v) / \epsilon\tau \rceil * \epsilon\tau$, i.e. the distance $\mathbf{dist}_{G^*}(u,v)$ rounded up to the nearest multiple of $\epsilon\tau$. This gives $2$-hops for all vertices that are sparse as described initially and some direct hops for vertices that are already sufficiently close to $s$ which are necessary to avoid a large additive error but should be seen as a technical detail.

\paragraph{Analyzing Stretch, Hop and Sparsity.} Let us offer some intuition about our clustering by sketching a proof by induction that establishes that there is a $(1+\epsilon)$-approximate shortest path from $s$ to every vertex $v \in V$ in each version of $H^s_{\mu, \tau, \epsilon}$ (plus a small additive error of $\epsilon\tau$). Let us therefore prove the claim that for every $\delta \geq 0$,
\begin{enumerate}
    \item we have for all vertices $v \in V$ at distance at most $\delta$ from $s$, that $\mathbf{dist}_{G^*}(s,v) \leq \mathbf{dist}_{H^s_{\mu, \tau, \epsilon}}(s,v) \leq (1+\epsilon)\mathbf{dist}_{G^*}(s,v) + \epsilon\tau$, and
    \item for every active center $c \in A_i$ at distance at most $\delta + 4\tau(1/\epsilon)^{i}$ from $s$, we have that $\mathbf{dist}_{H^s_{\mu, \tau, \epsilon}}(s,c) \leq \max\{(1+\epsilon)\mathbf{dist}_{G^*}(s,c) - 8\tau(1/\epsilon)^{i}, \mathbf{dist}_{G^*}(s,c)\}$.
\end{enumerate}
Clearly, for $\delta \leq 4\tau(1/\epsilon)^{k}$, we can simply use the edges in $\{s\} \times B_{G^*}(s, 2\tau(1/\epsilon)^{k+1})$ proving the base case (here we assume $\epsilon$ is sufficiently small). 

For any larger $\delta$, each \textit{active} center $c \in A_i$ can use an edge $(c,u)$ where $u$ is the vertex in $B_{G^*}(c, \tau(1/\epsilon)^{i+1})$ on the shortest path from $s$ to $c$ that is closest to $s$. By definition, $\mathbf{dist}_{G^*}(s,u) \leq \mathbf{dist}_{G^*}(s,c) - \tau(1/\epsilon)^{i+1} < \delta$ for $\epsilon \leq \frac{1}{10}$. Thus, we can invoke the induction hypothesis on $u$ and obtain that $\mathbf{dist}_{H^s_{\mu, \tau, \epsilon}}(s,c) \leq (1+\epsilon)\mathbf{dist}_{G^*}(s,u) + \epsilon\tau + \tau(1/\epsilon)^{i+1} \leq (1+\epsilon)\mathbf{dist}_{G^*}(s,c) - 8\tau(1/\epsilon)^{i}$ (where we again use $\epsilon \leq \frac{1}{10}$).

For each vertex $v \in V$, let $c$ be the closest vertex in $C_{\mu, \tau}$ (recall $\mathbf{dist}_G(v,c) \leq \tau$). Let $i = \textsc{ClusterLevel}_{\mu, \tau}(c)$, then either $c \in A_i$ or there is a vertex $c'$ in $A_i$ at distance at most $2\tau(1/\epsilon)^{i}$ from $c$ (by the maximality of $A_i$). Using the triangle inequality, we deduce that there is a vertex $c'' \in A_i$ at distance at most $2\tau(1/\epsilon)^{i} + 2\tau \leq 4\tau(1/\epsilon)^{i}$. Then using the edge to $c'$ (and by invoking the induction hypothesis on $c'$) establishes our claim. 

It is straight-forward to see from the exposed paths in our proof that we reduce the distance between $v$ and $s$ by at least $\tau$ using a 2-hop shortcut. Thus, the exposed path contains at most $2 \lceil\mathbf{dist}_{G^*}(s,v)/\tau \rceil$ edges, thus for distance $m\tau/\mu$, the number of edges is roughly $\tau/m$. To see that the emulator has at most $O(n^{1+o(1)})$ edges we point out that the careful construction of the sets $A_i$ ensures that every active center $c \in A_i$ has a set of at least $n^{i/k}$ vertices in its ball $B(c, \tau(1/\epsilon)^i)$ (which is disjoint from all balls at other vertices in $A_i$) and contributes by at most $n^{(i+1)/k}$ edges to $H^s_{\mu, \tau, \epsilon}$. Using the pigeonhole-principle and summing over the $k$ levels, we obtain the claimed bound.

\paragraph{Handling Dense Subgraphs.} The problem with this approach is that vertices that are not \textit{light} might not have \emph{any} incident edges in $H^s_{\mu, \tau, \epsilon}$ since \textit{heavy} vertices are not guaranteed to have a cluster center in the $\tau$-cover. However, since we chose the \textit{heavy}/\textit{light} threshold quite leniently, we even ensured that every vertex incident to a \textit{light} vertex has a 2-hop, i.e. has an edge to an active center of small weight that allows it to reach a vertex that is at least $\tau$ closer to $s$. Thus, we are left to expose shortcuts for \textit{heavy} vertices that are not incident to any \textit{light} vertex.

We can therefore use a simple trick introduced by Bernstein and Chechik \cite{bernstein2016deterministic}: We let $\textsc{Heavy}$ denote the set of vertices in $V$ that are \textit{heavy}. We add for each connected component $C$ in the induced graph $G^*[\textsc{Heavy}]$ a \textit{component} vertex $c$ to the vertex set $V_{\mu, \tau, \epsilon}$ of our emulator $H^s_{\mu, \tau, \epsilon}$ and add the edge $(c, v)$ with weight $\epsilon\tau$ to the edge set of $H^s_{\mu, \tau, \epsilon}$ for every vertex $v \in C$. 

Observe that this process adds at most $O(n)$ edges to our emulator. Moreover, every \textit{heavy} vertex $v$ can now use a single edge to get to its component vertex $c$. Clearly, $c$ has an incident edge to the vertex $u$ in the component $C$ that is closest to $s$ on the shortest path $\pi_{s,v}$. But since $u$ is incident to a \textit{light} vertex, we can take a $2$-hop from $u$ to some vertex $x$ that is at least $\tau$ closer to $s$ than $u$. The additive error of $2\tau\epsilon$ that we additionally accumulated to get from $v$ to $u$ can be subsumed in the multiplicative error of the $4$-hop to get from $v$ to $x$ by changing the constants. It is noteworthy that this might lead to underestimating the shortest path from $s$ to $v$ since the distance from $u$ to $v$ might be significantly larger than $2\epsilon\tau$. We therefore have to bound the induced underestimate of any shortest path but by using an advanced analysis, we can show that this incurs at most $O(\frac{\tau m}{\mu})$ (negative) additive error.

\paragraph{Maintaining the heavy/light decomposition.} We point out that the fix offered in the last paragraph requires us to monitor for every vertex $v \in V$, when it transitions from being \textit{heavy} to becoming \textit{light}. Unfortunately, maintaining this information takes $O(n\mu \tau(1/\epsilon)^{k+1})$ time which is too slow for our purposes. We therefore present a relaxation of the heavy/light decomposition that preserves most desirable properties whilst reducing the running time to $\tilde{O}(n\mu)$ where we again employ cover $C_{\mu, \tau}$. This allows us to scale our approach.

\paragraph{Running MES-trees on Threshold Emulators.} In order to use the threshold emulator to find distances from $s$, we maintain another layering with levels $i = [0, \lg n]$ where we want to maintain at each layer the distances $\mathbf{dist}_{G^*}(s,v) \in [2^i, 2^{i+1}]$. We therefore maintain a MES-tree rooted at $s$ to depth $\Theta(2^{i})$ on $H^s_{\mu, \tau_i, \epsilon}$ where we set $\tau_i = \frac{\mu 2^i \epsilon^{\Theta(k)}}{m}$. (Here we ignore the case when $\tau < 1$ since we can use a simple ES-tree from $s$ to find these distances exactly in time $O(m \frac{m}{\mu})$).

Observe, by preceding arguments, the emulator underestimates shortest paths by an additive error of $O(\frac{\mu \tau_i}{m}) = O(\epsilon 2^i)$. However, since we are only interested in distances that are greater-equal than $2^i$, the additive error can be subsumed in an $(1+\epsilon)$ multiplicative error (by adding $O(\epsilon 2^i)$ to each distance estimate). Thus, all relevant distances estimates maintained are $(1+\epsilon)$-approximate. 

Let us briefly sketch the running time of this approach. Following our previous arguments for $\tau$-covers, maintaining all $C_{\mu, \tau}$ for $\tau$ being powers of $2$ takes $\tilde{O}(n\mu)$ time. Maintaining for each $i$, a MES-tree on $H^s_{\mu, \tau_i, \epsilon}$ to depth $O(2^i)$ can be bounded by $O(n^{1+o(1)} \frac{m}{\mu})$ since we can divide edge weights by $\epsilon\tau_i$ in each instance which decreases the total update time by a fraction of $\epsilon\tau_i$. Running an exact ES-tree for small distances takes $O(\frac{m^{2 + o(1)}}{\mu})$. Thus, setting $\mu = \frac{m}{\sqrt{n}}$ optimizes the trade-off and gives $O(m n^{0.5+o(1)})$ total running time.

\paragraph{Using {almost}-MES-trees.} Whilst previous algorithms used MES-trees directly on the emulator, our emulators would induce high running time \emph{and} additive error if an MES-tree would be run naively on it. Whilst we can still preserve the time and stretch bounds stated, this requires new techniques to handle edge insertion/deletions in $H^s_{\mu, \tau_i, \epsilon}$.

To understand why difficulties arise in the first place, it is crucial to keep in mind that ES-trees/MES-trees scan the in-going edges of a vertex for every increase of their distance estimate once. Whilst our emulators have only $O(n^{1+o(1)})$ edges, this is a holistic bound taat does not bound the number of edges incident to a vertex. Thus, our emulator could place $\Theta(n)$ edges at a vertex $v$ and increase its distance estimate by $\Delta$. The induced costs would amount to $O(\Delta n)$. Then, the emulator could remove the edges from $v$ and insert them at some vertex $v'$ and repeat this scenario. 

To foil this scenario, we introduce a new \uline{dragging} technique where active centers also increase the distance estimates of vertices that are "close" to it. The fundamental idea is that since a vertex $v$ in the ball $B(c, \tau(1/\epsilon)^i)$ of some active center $c$ is "close" to $c$, then if $c$'s distance estimate increases by a lot (say more than $2\tau(1/\epsilon)^i$), we can also increase $v$'s distance estimate. Thus, the vertex $c$ \uline{drags} vertices in its ball along to amortize the costs of increasing its distance estimate due to being a high-degree vertex over distance estimate increases of many low-degree vertices in its ball $B(c, \tau(1/\epsilon)^i)$.

Finally, our emulator also requires some vertices to decrease their distance estimates occasionally to avoid incurring too much additive error. We show that the decreases do not affect the running time up to constant factors. For the reason of these occasional decreases, we call the data structure used instead of the MES-tree an \emph{almost}-MES-tree.

These two new techniques of dragging vertices along and allowing for occasional level decreases might be of independent interest and we hope that they are studied in further detail. However, we refrain from stating more details in the overview since the description of these techniques is rather technical.

\section{Background}
\label{subsec:basicDefs}

A dynamic graph $H$ is a sequence of graphs on the same vertex set $V(H)$ where the edge sets of two consecutive graph versions differ only by at most a single edge. We refer to $H$ at the stage $i$ to denote the graph version $H$ after the first $i$ updates, thus $H$ at stage $0$ is the initial graph. To denote $H$ at the current stage (i.e. the stage under discussion) we often omit stating the version explicitly. In this paper, if not explicitly pointed out, we are only concerned with undirected, unweighted graphs. If for all $i > 0$, the graph $H$ at stage $i$ is a subgraph of the graph $H$ at stage $i-1$, we say that $H$ is a decremental graph. Throughout the paper, we assume that any decremental graph $H$ is initially fully-connected.

For the decremental input graph $G$, we let $m$ denote the cardinality of the initial edge set, i.e. the edge set of $G$ at stage $0$, and $n$ refer to the cardinality of its vertex set.

For $u,v \in V$, we let $\mathbf{conn}(u,v)$ be true if and only if there is a path from $u$ to $v$ in graph $G$. We henceforth denote by $G^*$ the decremental graph $G$ where edge deletions resulting in an increased number of connected components in $G$ are skipped. Therefore $G^*$ is fully-connected at every stage. Using a Deterministic Connectivity Data Structure on $G$ (for example  \cite{holm2001poly, Wulff-Nilsen:2013:FDF:2627817.2627943}), we can maintain $G^*$ introducing an overhead of $O(m \log^2 n)$ total time as follows: we pick an arbitrary vertex $r \in V$ and assign each vertex $x \in V$ a component representative $r(x)$ which is initially set to $r$. If we discover that the deletion of an edge $(u,v)$ disconnects $u$ from $v$, we query the Connectivity Data Structure for the size of the components $C_u$ and $C_v$ with $u$ contained in $C_u$ and $v$ in $C_v$. Let w.l.o.g. $C_u$ be the smaller component, then we assign each vertex $x \in C_u$ the new representative $r(x) = u$. This takes only $O(n \log{n})$ overall time since each time a component needs to be scanned, the size of the component halves at least which is subsumed by the update time incurred by the Connectivity Data Structure into account. We can then query $\mathbf{conn}(u,v)$ at any stage in constant time by checking if $r(u) = r(v)$.

We let $\dist[H]{u,v}$ denote the shortest path distance from $u$ to $v$ in graph $H$ and let $\pi_{u,v}$ denote a shortest path between $u$ and $v$ (or some consistently chosen shortest path if there are more than a single one). We sometime say the \textit{length} or \textit{weight} of a path $\pi$ to denote the sum of all edge weights of edges on $\pi$. We argue that for the decremental graph $G$, we have for all stages $i$ where $u,v \in V$ are connected in $G$ that $\dist[G]{u,v} = \dist[G^{*}]{u,v}$ and otherwise $\dist[G]{u,v} = \infty$. Thus, $\dist[G]{u,v} \leq \dist[G^{*}]{u,v}$. Moreover, for any vertex $c \in V$ and given $u,v \in V$ where $u$ and $v$ are connected in $G$, we have that if $\dist[G^{*}]{u,c} \leq d$ for any $d \in [0,n]$, then $\dist[G^{*}]{v,c} \leq \dist[G^{*}]{u,v} + \dist[G^{*}]{u,c} \leq \dist[G]{u,v} + d$. This holds even if $c$ is not in the same component as $u$ or $v$ in $G$.

We further define for any vertex $v \in V$ in the decremental graph ${G}$, the ball 
\[
B_G(v, r) = \{u \in V | \dist[G]{u,v} \leq r\},
\]
where we drop the subscript when the context is clear. We let $\mathbf{deg}(B_{G^*}(v,r))$ denote the number of edges with both endpoints in the set $B_{G^*}(v,r)$. This allows us to define a simplified version of the heavy/light decomposition that was used before in  \cite{bernstein2017deterministic}. 

\begin{definition}
\label{def:heavy}
For any $\mu > 0$, we say a vertex $v \in V$ is $(\mu, r)$-\textit{heavy} if $\mathbf{deg}(B_{G^*}(v,r)) > \mu$, otherwise we say $v$ is $(\mu, r)$-\textit{light}. 
\end{definition}

Observe that our definition is in regard to ${G}^*$, not $G$. It is straight-forward to extend the ES-tree to detect the transition of a vertex from being $(\mu, r)$-\textit{heavy} to being $(\mu, r)$-\textit{light} and to extend the data structure to support some additional operations.

\begin{lemma} [c.f.  \cite{bernstein2017deterministic}]
\label{lma:maintainBalls}
For any vertex $s \in V$, we can maintain a data structure on a decremental graph $G^*$ that can report when $s$ transitions from $(\mu, r)$-\textit{heavy} to $(\mu, r)$-\textit{light} and once $s$ is $(\mu, r)$-\textit{light} support the following operations:
\begin{itemize}
    \item It can report when a vertex $w \in B(s,r)$ has increased its distance and be queried for any distance $\mathbf{dist}(s,w)$ for $w \in V$.
    \item Initially given a partition $\mathcal{P}$ of $[0,r]$, it can maintain for each set $P \in \mathcal{P}$ the set of vertices at one of the distance in $P$, i.e. the set
    \[
        V(P) = \{ \mathbf{dist}(s,v) \in P | v \in V\}
    \]
    and report any changes to $P$.
    \item It can support updates to \textit{mark} or \textit{unmark} vertices in $B(s,r)$ and support queries on whether there is a marked vertex in $B(s,r)$.
\end{itemize}
For $\Delta$ updates and queries to the data structure, the total time for updates, queries and reporting changes is $O(\mu r + \Delta)$. Further, the data structure can return a shortest path $\pi_{s,t}$ for any $t \in B(s,r)$ in $O(|\pi_{s,t}|)$ time. Observe that $\pi_{s,t}$ is a shortest path in $G^*$. 
\end{lemma}

This data structure can be implemented by a straight-forward extension of \emph{ES}-trees as introduced in  \cite{shiloach1981line}. ES-trees internally maintain a shortest-path tree $T$ rooted at a vertex $s \in V$ and maintain for each vertex $v$ a distance estimate $\hat{l}(v)$ that reflects the weight of the $s$ to $v$ path in $T$. The ES-tree then handles edge deletions by updating distance estimates after each edge deletion to 
\[
    \hat{l}^{NEW}(v) = \min_{(u,v) \in E} \hat{l}^{NEW}(u) + w(u,v).
\]
The ES-tree gives a clever implementation to update the distance estimates efficiently. Therefore it maintains a queue $Q$ over all vertices whose parent edge $(u,v)$ in $T$ does not satisfy the above-stated update rule. It then iteratively removes the vertex $v$ with smallest distance estimate in $Q$, checks its in-going edges on whether any edge satisfies the above-stated rule and inserts the first edge that satisfies the equality to $T$ and removes $v$ from $Q$. If no such edge exists the distance estimate of $v$ is increased by one (which might add descendents of $v$ to $Q$). 

To gain efficiency, the ES-tree implementation does not actually look at all in-going edges of a vertex $v$ every time its tree edge does not satisfy the above-stated property. Instead, once it checked an in-going edge for a certain value of $\hat{l}(v)$ and learned that it does not satisfy the equation, it does not look at this edge again until the distance estimate $\hat{l}(v)$ increased. Thus, each vertex $v$ only looks at every value $\hat{l}(v)$ only once at its in-going edges and therefore the total running time to maintain $T$ to depth $r$ is $O(\sum_{v\in V} \mathbf{deg}_G(v) r) = O(mr)$ time.

Whilst the ES-tree only works on decremental graphs, the MES-tree can maintain distances in dynamic graphs if they satisfy certain properties. The main idea behind MES-trees is that distance estimates are monotone, i.e. even if a newly-inserted edge would decrease the distance between two vertices, the distance estimate is not decreased. More formally, after each edge deletion, the rule
\[
    \hat{l}^{NEW}(v) = \max\{ \hat{l}(v), \min_{(u,v) \in E} \hat{l}^{NEW}(u) + w(u,v)\}
\]
is enforced.

\section{A layered graph cover}
\label{sec:cover}

In this section, we show how to efficiently maintain a graph cover on the decremental graph $G^*$. We start by defining a layered $\tau$-cover. 

\begin{definition}
\label{def:coverSimple}
We call $C \subseteq V$ a $\tau$-cover of $G^
*$, where $\tau \geq 2$, if for every vertex $v \in V$, at any stage, $\min_{c \in C} \mathbf{dist}_G(v, c) \leq \tau$ and let $p_C(v)$ be a vertex $c \in C$ at distance at most $\tau$ to $v$.

We call a collection $\mathcal{C} = \{C_0, C_1, .., C_{\lfloor\log_{\tau} n\rfloor}\}$ a layered $\tau$-cover, if $C_j$ for every $j \in [0, \lfloor\log_{\tau} n\rfloor]$, is a $\tau^j$-cover. 
\end{definition}

Consider the following simple greedy algorithm to maintain a layered $\tau$-cover: We initially let $C_0 = V$. Then for any $j$, we say a vertex $v \in C_{j}$ is $(j+1)$-\textit{eligible} if no other vertex in $B(v, \tau^{j})$ is already in $C_{j}$.

While there is a vertex $v$ that is $(j+1)$-eligible at any point of the algorithm, we add it to $C_{j+1}$ and then check whether any other vertex is still $(j+1)$-eligible. We break ties arbitrarily when processing $(j+1)$-eligible vertices. We use ES-trees as described in lemma \ref{lma:maintainBalls} with $\mu = m$, from each vertex $c$ added to $C_j$ to depth $\tau^j$ and mark the vertices in $B(c, \tau^j)$ that are already in $C_{j+1}$. Then, if the number of marked vertices in the ES-tree is zero, the vertex $c$ is $(j+1)$-eligible.

\begin{lemma}
\label{lma:coverefficientSimple}
Using the algorithm given above, we can maintain a layered $\tau$-cover $\mathcal{C}= \{C_0, C_1, .., C_{\lfloor\log_{\tau}n\rfloor}\}$ with each $C_j$ being a monotonically increasing set over time and satisfying $|C_j| \leq \frac{8n}{\tau^j}$ for each $j \in [0,\lfloor\log_{\tau} n\rfloor]$ in total time $O(nm \log n)$. We can maintain within the same time $p_{C_j}(v)$ for every $j \in  [0,\lfloor\log_{\tau} n\rfloor]$ and $v \in V$.
\end{lemma}
\begin{proof}
We prove by induction on $j$. Since $C_0 = V$, the lemma is vacuously true in the base case.

Let us take the induction step for $j > 0$. To show that for each vertex $v \in V$, there exists a vertex $c \in C_j$ at distance at most $\tau^j$ at each stage, we first invoke the induction hypothesis, to conclude that there is a vertex $c'$ in $C_{j-1}$ at distance at most $\tau^{j-1}$ from $v$. Observe that by our algorithm, either $c'$ was later added to $C_j$ or there exists a vertex $c'' \in C_j$ at distance at most $\tau^{j-1}$. In either case, since $\tau \geq 2$, this ensures that there is a vertex in $C_j$ that is at distance at most $\tau^j$ to $v$.

Let us now analyze the number of vertices that are in a set $C_{j}$. Any two vertices $u,v \in C_{j}$ are at distance $\dist[G^*]{u,v} > \tau^{j-1}$, thus the ball $B(u, \tau^{j-1}/2)$ is disjoint from all balls $B(v, \tau^{j-1}/2)$ for every $u,v \in V, v \neq u$. Since $G^*$ is connected by definition, we have that $|B(v, \tau^{j-1}/2)| \geq \tau^{j-1}/2$ and by simple counting arguments, there can be at most $2n/\tau^{j-1}$ such disjoint balls and therefore vertices in $C_j$. 

Let us finally bound the running time of the algorithm. For each level $j \in [0, \lfloor\lg_{\tau}(n)\rfloor]$, we run by the preceding analysis an ES-tree as described in \ref{lma:maintainBalls} from $O(n/\tau^{j-1})$ vertices to depth $r = \tau^j$. Each ES-tree runs in $O(m \tau^j)$ time, thus the total update time is $O(mn \tau \log n)$. We observe that marking the vertices in other ES-trees once a vertex is added to $C_j$ can be done during the initial shortest path tree computation for the ES-tree (i.e. the computation of a BFS to depth $\tau^{j+1}$). To maintain $p_{C_j}(v) \in C_j$, we can ask each ES-tree at level $j$, to register as such at vertex $v$ if that vertex is contained in the tree and inform it in case of distance increases. 

Finally, we point out that we can always run a layered $2$-cover $\mathcal{C}$, and for any $\tau \geq 2$ map $j \in [0, \lfloor\lg_{\tau}(n)\rfloor]$ to the set $C_{j'}$ with $j' = \lfloor j * \lg(\tau) \rfloor$. The guarantees and update time follow.
\end{proof}

Let us know show how to use such a layered $\tau$-cover $\mathcal{C}$ to give an efficient implementation of decremental $(1+\epsilon)$-approximate APSP, establishing one of our main results.

\begin{corollary}[Restatement of Theorem \ref{thm:mainAPSP}]
\label{cor:APSPfastAndSimple}
We can maintain $(1 + \epsilon)$-approximate all-pairs shortest paths for an undirected, unweighted dynamic graph $G=(V,E)$ and any $0 <\epsilon < 1$ in $O(mn \log{n}/\epsilon)$ time, where distance estimate queries for any vertices $u,v \in V$ take $O(\log\log n)$ time.
\end{corollary}
\begin{proof}
We simply run the algorithm above to maintain a layered $2$-cover on $G^*$ which takes time $O(mn \log{n})$, as described in lemma \ref{lma:coverefficientSimple}. Further, we run from each vertex $c \in C_j$ and every $j \in [0, \lfloor \lg n \rfloor]$, an ES-tree as described in lemma \ref{lma:maintainBalls} with $\mu = m$ and depth $r = b2^j$, for some $b = \Theta(1/\epsilon)$ to be specified later, summing to total update time $O(mn \log n /\epsilon)$.

Now, consider a query for the distance estimate from $u$ to $v$. We first check in constant time whether $r(u) = r(v)$ and if not we return $\infty$. Otherwise, we find the smallest $j \in [0, \log n]$ with $v \in B(p_{C_j}(u), b2^j)$ and return the distance estimate $\adist{u,v} = \dist[G^*]{u, p_{C_j}(u)} + \dist[G^*]{p_{C_j}(u), v}$. 

Let $j$ be the smallest found index. We start the stretch analysis by lower bounding $\dist[G]{u,v}$. Since $v \not\in B(p_{j-1}(u), b2^{j-1})$, we have that $\dist[G^*]{p_{j-1}(u),v} \geq b 2^{j-1}$. Additionally, we can argue by applying the triangle inequality that $\dist[G]{u,v} + \dist[G^*]{p_{j-1}(u),u} \geq \dist[G^*]{p_{j-1}(u),v}$.
Combining the preceding two inequalities with $\dist[G^*]{p_{j-1}(u),u} \leq 2^{j-1}$, we derive that \[
\dist[G]{u,v} \geq \dist[G^*]{p_{j-1}(u),v} - \dist[G^*]{p_{j-1}(u),u} \geq (b-1) 2^{j-1}
\]
and for $b = 1 + 4/\epsilon$, we have $\dist[G]{u,v} \geq 2^{j+1} /\epsilon$. Assuming that $u$ and $v$ are connected, we conclude,
\begin{equation*} 
\begin{split}
\adist{u,v} &= \dist[G^*]{u, p_{C_j}(u)} + \dist[G^*]{p_{C_j}(u), v} \leq 2\dist[G^*]{u, p_{C_j}(u)} + \dist[G^*]{u, v} \\
&\leq 2^{j+1} + \dist[G^*]{u, v} \leq (1+\epsilon)\dist[G^*]{u, v}
\end{split}
\end{equation*}
where we use the triangle inequality again for the first inequality. As argued in the preliminaries, if $u$ and $v$ are in the same component of $G$, then $\dist[{G}]{u,v} = \dist[{G}^*]{u,v}$ and it is therefore straight-forward to conclude $\dist[{G}]{u,v} \leq \adist[{G}]{u,v}$.

Checking an index $j$ on whether $v \in B(p_{C_j}(u), b2^ j)$ takes constant time. Instead of checking every index, we apply a binary search over the indices as follows: We take the index that halves the current range of potential indices. We then check if $v$ is in $B(p_{C_j}(u), b2^j)$ and if $v \in B(p_{C_j}(u), b2^j)$, we can reduce our search to all indices less than equal to $j$ and otherwise to indices greater than $j$. Thus, finding the correct index takes $O(\log \log n)$ time.

To see that we return the correct index, let us first show that if $v \in B(p_{j}(u), b2^{j})$ then $v \in B(p_{j+1}(u), b2^{j+1})$ for all $j$. This follows since we have $\mathbf{dist}(p_{j}(u), v) \leq b2^j$, $\dist{p_{C_j}(u), u} \leq 2^j$ and $\dist{p_{j+1}(u), u} \leq 2^{j+1}$ and therefore by the triangle inequality, $\dist{p_{j+1}(u), v} \leq \mathbf{dist}(p_{j}(u), v) + \dist{p_{C_j}(u), u} + \dist{p_{j+1}(u), u} \leq (b+3)2^{j}$. 
But $(b+3)2^j < b2^{j+1}$ since $b \geq 5$.

We are left to show that if  $v \not\in B(p_{j+1}(u), b2^{j+1})$ then $v \not\in B(p_{j}(u), b2^{j})$ for all $j$. Observe therefore that $\dist{p_{j+1}(u), v} > b2^{j+1}$, $\dist{p_{j+1}(u), u} \leq 2^{j+1}$ and $\dist{p_{C_j}(u) , u} \leq 2^j$. By the triangle inequality, we therefore have $\dist{p_{j}(u), v} \geq \dist{p_{j+1}(u), v} - \dist{p_{j+1}(u), u} - \dist{p_{C_j}(u) , u} \geq (2b-3)2^j$. Again, since $b \geq 5$, $(2b-3)2^j > b2^j$ and therefore $v \not\in B(p_{C_j}(u), b2^{j})$.

To attain a path from $u,v$ of length at most $\adist{u,v}$, we use the ES-trees which can also return shortest paths for each distance estimate in linear time in the number of vertices on the path, but point out that these shortest paths are found in $G^*$. For the estimate $\dist[{G}^*]{u, p_{C_j}(u)} + \dist[{G}^*]{p_{C_j}(u), v}$, we find the path in from $u$ to $p_{C_j}(u)$ but stop after finding the first endpoint incident to a marked edge in the top tree. We then follow the path from $v$ to $p_{C_j}(u)$ in the same manner. Since $u$ and $v$ are connected in $G$, they are in the same component and the paths to $p_{C_j}(u)$ both use the same marked edge to reach $p_{C_j}(u)$. Stopping before this edge is equivalent to removing a cycle. Thus, we can retract a valid path in ${G}$ in linear time.
\end{proof}

Finally, let us consider a generalization of the layered $\tau$-cover, that we use in the subsequent sections.

\begin{definition}
\label{def:cover}
We call $C \subseteq V$ a $(\mu, \tau)$-cover of $G^
*$, where $\tau \geq 2$ if for any $(\mu, 2\tau)$-\textit{light} vertex $v \in V$ there is a $c \in C$ such that $\mathbf{dist}(v,c) \leq \tau$. We call the collection $\mathcal{C}= \{C_0, C_1, .., C_{\lfloor\log_{\tau}(n)\rfloor}\}$ a \textit{layered $(\mu , \tau)$-cover} if each $C_j \in \mathcal{C}$ is a $(\mu, \tau^j)$-cover. Again, we denote by $p_{C_j}(v)$ some arbitrary vertex in $C_j$ that is at distance at most $\tau^j$ from $v$ and let $p_{C_j}(v) = \bot$ if no such vertex exists.
\end{definition}

In particular, we adapted the layered cover definition to only ensure for light vertices to be covered, but do not require this for heavy vertices. This enables us to implement the layered cover more efficiently. Additionally, we introduce a parameter $\gamma$ that allows us to start our layered cover at a higher layer. This parameter is defined for convenience and will ease the description of the hopset algorithm in the next section.

We point out at this point that replacing the criterion of $(j+1)$-eligibility in the preceding greedy algorithm, by defining a vertex $v \in C_j$ to be $(j+1)$-\textit{eligible} if
\begin{itemize}
    \item $v$ is $(\mu, 2\tau^{j+1})$-\textit{light}, and
    \item no other vertex in $B(v, \tau^{j})$ is already in $C_{j}$.
\end{itemize}
and using the ES-trees with parameter $\mu$ gives a straight-forward algorithm to maintain the generalized layered covers. We state the result of this improved analysis in the lemma below but defer the proof to appendix \ref{sec:appendixCoverProof}.

\begin{lemma}
\label{lma:efficientCover}
Using the algorithm given above, we can maintain a layered $(\mu, \tau)$-cover $\mathcal{C} = \{C_0, C_1, .., C_{\lfloor\log_{\tau}(n)\rfloor}\}$ with $C_j \subseteq V$ being a monotonically increasing set over time and satisfying $|C_j| \leq \frac{8n}{\tau^j}$ for any $\tau \in [0,n], j \in [0,{\lfloor\log_{\tau}(n)\rfloor}]$ in total time $O(n\mu \log n)$. We can maintain within the same time $p_{C_j}(v)$ for every $j \in [0,\lfloor\log_{\tau} n/\gamma\rfloor]$ and $v \in V$.
\end{lemma}

Throughout the rest of the paper, we let $\mathcal{C}_{\mu} = \{C_{\mu, 0}, C_{\mu, 1}, .., C_{\mu, \lfloor\lg n\rfloor}\}$ denote a specific instance of a layered $(\mu, 2)$-cover.

\section{The Threshold Emulator}
\label{sec:emulator}

In this section, we define the emulator $H^s_{\mu, \tau, \epsilon}$ more precisely and present efficient algorithms to maintain it. Therefore, we first introduce the important notion of a near-heavy/near-light decomposition and show how to detect dense subgraphs. We further show that contracting all connected components consisting of near-heavy vertices only results in an additive error for any distance in the graph.

In section \ref{subsec:efficientclustering}, we then present our clustering approach and show how to construct and maintain the emulator $H^s_{\mu, \tau, \epsilon}$ efficiently. Whilst a proof sketch in the overview established various properties of $H^s_{\mu, \tau, \epsilon}$, we refrain from presenting proofs in this section. Instead, we analyze the properties directly in their interplay with the almost-MES which is tailored towards our specific emulator.

\subsection{Detecting Dense Subgraphs}
\label{subsec:detectsubgraphs}

As mentioned in the overview, maintaining when vertices transition from being $(\mu, r)$-\textit{heavy} to $(\mu, r)$-\textit{light} is rather expensive since it requires using an ES-tree as described in lemma \ref{lma:maintainBalls} for each vertex in $V$ which incurs total running time $O(n \mu r)$. Let us therefore introduce a related concept that will subsequently allow us to remove the dependency in $r$ in the running time.

\begin{definition}
\label{def:nearheavy}
For every vertex $v \in V$, we say that $v$ is $(\mu, r, \tau)$-\textit{near-light} at stage $j$, if at any stage $j' \leq j$, $p_{C_{\mu, \tau}}(v) \neq \bot$ and $p_C(v)$ was $(\mu, r + 2\tau)$-\textit{light}. Otherwise, we say that $v$ is $(\mu, r, \tau)$-\textit{near-heavy}.
\end{definition}

We point out that our definition implies that any vertex $v$ can transition only once from being $(\mu, r, \tau)$-\textit{near-heavy} to becoming $(\mu, r, \tau)$-\textit{near-light}. This property simplifies working with the definition and imitates properties of the original heavy/light decomposition. To gain some intuition for this definition, let us denote by $\textsc{Near-Heavy}(\mu, r, \tau)$ the set of all $(\mu, r, \tau)$-\textit{near-heavy} vertices in $G^*$ (recall that we defined the set $\textsc{Heavy}(\mu, r)$ quite similarly) and prove that every vertex that is $(\mu, r, \tau)$-\textit{near-heavy} is also $(\mu, r)$-\textit{heavy}.

\begin{claim}
We have at all stages that $\textsc{Near-Heavy}(\mu, r, \tau) \subseteq \textsc{Heavy}(\mu, r)$. 
\end{claim}
\begin{proof}
Consider, for the sake of contradiction, that at any stage $i$, we have a vertex $v \in \textsc{Heavy}(\mu, r) \setminus \textsc{Near-Heavy}(\mu, r, \tau)$. Then by definition \ref{def:nearheavy}, at some stage $j \leq i$, there was a vertex $c = p_{C_{\mu, \tau}}(v)$ that was $(\mu, r + 2\tau )$-\textit{light} which implies by definition \ref{def:heavy}, that $\mathbf{deg}(B_{G^*}(c, 2r + 2 \tau)) \leq \mu$.

Further, by lemma \ref{lma:coverefficientSimple}, we have that at stage $j$, $c$ was at distance at most $\tau$ from $v$. But using this fact in conjunction with the triangle inequality implies that $B_{G^*}(v, r) \subseteq B_{G^*}(c, r + 2 \tau)$ at stage $j$. Finally, since the number of edges in a ball is monotonically decreasing over time in a decremental graphs, we conclude that at stage $i$, $\mathbf{deg}(B_{G^*}(v, r) \leq \mu$ which in turn implies that $v$ is $(\mu, r)$-\textit{light} contradicting our assumption.
\end{proof}

Next, let $\hat{G}$ denote the graph $G^*$ where the vertex sets that form connected components in the graph $G^*[\textsc{Near-Heavy}(\mu, r, \tau)]$ are contracted. We call $V(\hat{G})$ the \textit{node} set to emphasize that it differs from the vertex set of $G^*$. We further introduce the convention that we let $X^v \in V(\hat{G})$ denote the node in $\hat{G}^*$ into which the vertex $v \in V$ was contracted. Let us now show that distances are underestimated by only a reasonable additive error when we contract these vertices.

\begin{claim}
\label{clm:guaranteeContractions}
At any stage, for any vertices $u, v \in V$, and $r, \tau > 0$,
\[
    \mathbf{dist}_{G^*}(u,v) \leq \mathbf{dist}_{\hat{G}}(X^u,X^v) + \frac{8(r + \tau) m}{\mu}
\]
\end{claim}
\begin{proof}
Consider some shortest path $\pi_{u,v}$ from $u$ to $v$ in $G^*$. Let us bound the number of vertices of vertices on $\pi_{u,v}$ that are $(\mu, r, \tau)$-\textit{near-heavy}. Clearly, this number upper bounds the additive error.

Pick therefore a maximal collection $\mathcal{C}$ of $(\mu, r, \tau)$-\textit{near-heavy} that are at pairwise distance of at least $4(r + \tau)$. By choice of our collection, the balls $B(c, r + 2 \tau)$ are pairwise disjoint (in fact even the balls $B(c, 2(r + \tau)-1)$ are pairwise disjoin). 

We claim that for each $c \in \mathcal{C}$, $\mathbf{deg}(B_{G^*}(c, r + 2\tau)) > \mu$. Consider for the sake of contradiction that this is not true. Then by lemma \ref{lma:coverefficientSimple}, there exists a vertex $p_{C_{\mu, \tau}}(c)$ in $C_{\mu, \tau}$ at distance at most $\tau$. Clearly, this vertex is $(\mu, r)$-\textit{light}, since $B_{G^*}(p_{C_{\mu, \tau}}(c), r) \subseteq B_{G^*}(c, r + 2\tau)$. But this contradicts that $c$ is $(\mu, r, \tau)$-\textit{near-heavy} by definition \ref{def:nearheavy}. 

Since each $\mathbf{deg}(B_{G^*}(c, r + 2\tau)) > \mu$ and by the disjointness of these balls, there can be at most $m/\mu$ vertices in $\mathcal{C}$. Finally, for each vertex $c \in \mathcal{C}$, there are at most $8(r + \tau)$ vertices on $\pi_{u,v}$ that are $(\mu, r, \tau)$-\textit{near-heavy} (the vertices in $\pi_{u,v} \cap B_{G^*}(c, r + 2\tau)$), thus there are at most $\frac{8(r + \tau) m}{\mu}$ vertices that are $(\mu, r, \tau)$-\textit{near-heavy} on $\pi_{u,v}$, as required. 
\end{proof}

Finally, let us prove that monitoring near-heaviness/near-lightness and maintaining the connected components of $G^*[\textsc{Near-Heavy}(\mu, r, \tau)]$ can indeed be implemented efficiently. 

\begin{claim}
\label{clm:maintainNearHeavy}
For every vertex $v \in V$, we can monitor when $v$ transitions from being $(\mu, r, \tau)$-\textit{near-heavy} to becoming $(\mu, r, \tau)$-\textit{near-light} in $O(n\mu r/\tau)$ total update time.
\end{claim}
\begin{proof}
We use the ES-tree data structure presented in lemma \ref{lma:maintainBalls} on every vertex in $C_{\mu, \tau}$. It is then straight-forward to implement the operation by checking every time that $p_{C_{\mu, \tau}}(v)$ is changed whether $v$ is already $(\mu, r, \tau)$-\textit{near-light} and otherwise by reporting that $v$ transitioned.
\end{proof}

\begin{claim}
\label{clm:heavysubgraphmaintain}
We can maintain a data structure that monitors the connected components in $G^*[\textsc{Near-Heavy}(\mu, r, \tau)]$ such that every time a connected component $C$ decomposes into components $C_1$ and $C_2$, it returns a list of the vertices of $C_1$ where $|C_1| \leq |C_2|$. The data structure requires $O(n\mu r/\tau + m \log^2 n/ \log\log n)$ total update time.  
\end{claim}
\begin{proof}
It is straight-forward to maintain the set $\textsc{Near-Heavy}(\mu, r, \tau)$ using claim \ref{clm:maintainNearHeavy}. We further observe that since the vertex set $\textsc{Near-Heavy}(\mu, r, \tau)$ is monotonically decreasing over time by the definition of the near-heavy/near-light decomposition. Therefore, we can model the data structure by running a decremental connectivity data structure on $G^*[\textsc{Near-Heavy}(\mu, r, \tau)]$. We can then forward edge deletions, and simulate transitions of vertices from near-heaviness to near-lightness by removing their incident edges and finally the vertex itself from the data structure. The data structure by Wulff-Nilsen  \cite{Wulff-Nilsen:2013:FDF:2627817.2627943} can fully support our claimed monitoring operation and establishes the update time. 
\end{proof}

\subsection{Efficient Clustering in Sparse Subgraphs}
\label{subsec:efficientclustering}

We take a layered approach at clustering using centers in $C_{\mu, \tau}$ as candidates for centers of clusters (we refer to selected candidates as \textit{active} centers). Throughout the algorithm, we maintain for each vertex in $C_{\mu, \tau}$ a level $i = \textsc{ClusterLevel}_{\mu, \tau}(c)$ in the hierarchy. Before we proceed our discussion, let us introduce some additional notation to ease the further discussion.

\begin{definition}
\label{def:coreCluster}
For any $v \in C_{\mu, r}$, at level $i = \textsc{ClusterLevel}_{\mu, \tau}(c)$, let us define the sets $\textsc{Cluster}_{\mu, \tau}(c)$ to be the set of vertices in the ball $B_{G^*}(c, \tau(1/\epsilon)^{i+1})$, and $\textsc{Core}_{\mu, \tau}(c)$ to be the set of vertices in the ball $B_{G^*}(c, \tau(1/\epsilon)^{i})$.
\end{definition}

In order to enforce sparseness of our emulator, we want to maintain $\textsc{ClusterLevel}_{\mu, \tau}(c)$ such that at all times $|\textsc{Cluster}_{\mu, \tau}(c)| \leq n^{1/k}  |\textsc{Core}_{\mu, \tau}(c)|$. This ensures that we provide many vertices (the vertices in $\textsc{Core}_{\mu, \tau}(c)$) a good 2-hop using relatively few edges in the emulator (namely at most an additional fraction of $n^{1/k} = n^{o(1)}$). Additionally, we want to maintain the level of each $c$ to be monotonically decreasing such that we do not have to change the clustering too many times since it is rather expensive to change cluster centers. Let us now prove that these properties can be enforced simultaneously.

\begin{claim}
\label{clm:maintainlevels}
We can maintain for any $(\mu, \tau(1/\epsilon)^{k+1})$-\textit{light} vertex $c \in C_{\mu, \tau}$ an index $i = \textsc{ClusterLevel}_{\mu, \tau}(c) \in [0, k]$ with the following properties:
\begin{enumerate}
    \item $\textsc{ClusterLevel}_{\mu, \tau}(c)$ is monotonically decreasing over time, and
    \item $n^{i/k} \leq |\textsc{Core}_{\mu, \tau}(c)| \leq  |\textsc{Cluster}_{\mu, \tau}(c)| \leq n^{(i+1)/k}$. \label{prop:sparse}
\end{enumerate}
over all stages in $O(\mu \tau(1/\epsilon)^{k+1})$ time.
\end{claim}
\begin{proof}
Let $\mathcal{P} = \{ [0, \tau(1/\epsilon)^{0}), (\tau(1/\epsilon)^{0}, \tau(1/\epsilon)^{1}], .. , (\tau(1/\epsilon)^{k}, \tau(1/\epsilon)^{k+1}] \}$ be a partition of the distances $[0, \tau(1/\epsilon)^{k+1}]$. We use an ES-tree as described in lemma \ref{lma:maintainBalls} with $\mathcal{P}$ to depth $\tau(1/\epsilon)^{k+1}$ to maintain for each interval in $\mathcal{P}$, the number of vertices at distance in the interval (it is here that we use that $c$ is $(\mu, \tau(1/\epsilon)^{k+1})$-\textit{light} by assumption).

Using this information, we maintain $\textsc{ClusterLevel}_{\mu, \tau}(c)$ to be the largest integer $i$ such that for all $j \leq i$, 
\[
|B_{G^*}(c, \tau(1/\epsilon)^j)| \geq n^{j/k}.
\]
To see that $\textsc{ClusterLevel}_{\mu, \tau}(c)$ is decreasing observe that once $\textsc{ClusterLevel}_{\mu, \tau}(c)$ was set to a level $i$, we certified that $|B_{G^*}(c, \tau(1/\epsilon)^{i+1})| < n^{(i+1)/k}$ for the rest of the algorithm. This follows because balls are shrinking in size since distances are monotonically increasing in the underlying decremental graph. 

Finally, let us prove property \ref{prop:sparse}. Clearly, if the first inequality is always preserved by our choice of $i$ and the second inequality follows from $\textsc{Core}_{\mu, \tau}(c) \subseteq \textsc{Cluster}_{\mu, \tau}(c)$. The third inequality is again preserved by our choice of $i$ and the observation that in decremental graphs, balls only decrease in size.

The running time of the algorithm follows from lemma \ref{lma:maintainBalls}. 
\end{proof}

For convenience, we define $\textsc{ClusterLevel}_{\mu, \tau}(c) = \infty$ if a vertex $c$ in $C_{\mu, \tau}$ is $(\mu, 2\tau(1/\epsilon)^{k+1})$-\textit{heavy} (observe that double the radius used to determine heaviness due to technical reasons arising in the clustering algorithm). 

Let us now restate the algorithm to select the centers of our hierarchy and analyze it more formally: for each level $i \in [0,k]$, we maintain a set of \textit{active} centers $A_i$ that is a subset of the vertices in $C_{\mu, \tau}$ at cluster level $i$ such that no two vertices in $A_i$ are at distance less than $2\tau(1/\epsilon)^{i}$. In particular, the balls $B_{G^*}(c, \tau(1/\epsilon)^{i})$ for all $c \in A_i$ are pairwise disjoint.

\begin{claim}
\label{clm:maintainA}
We can maintain the dynamic sets $A_0, A_1, .., A_k$ such that 
\begin{enumerate}
    \item each $A_i$ is a maximal subset of vertices in $C_{\mu, \tau}$ at level $i$ such that two any vertices in $A_i$ are at distance at least $2\tau(1/\epsilon)^i$, \emph{and} \label{prop:Aisalwaysfilled}
    \item once a center $c \in C_{\mu, \tau}$ leaves a set $A_i$, it never joins it again. \label{prop:leaveA}.
\end{enumerate}
We can maintain these properties in time $O(n\mu (1/\epsilon)^{k+1})$.
\end{claim}
\begin{proof}
Our proof is quite similar to the proof of our cover data structure (lemma \ref{lma:coverefficientSimple}). Let $c \in C_{\mu, \tau}$ be assigned a new level $i$ (or if $\textsc{ClusterLevel}_{\mu, \tau}(c)$ is set for the first time to a value $\neq \infty$). Then, initiate a new ES-tree as described in lemma \ref{lma:maintainBalls} for each center $c \in C_{\mu, \tau}$ with depth $2(1/\epsilon)^i$ and initially mark all vertices in $B(c, 2(1/\epsilon)^i)$ that are in $A_i$. If at any stage there is no marked vertex in $B(c, 2(1/\epsilon)^i)$, then add $c$ to $A_i$ and run a BFS exploring $B(c, 2(1/\epsilon)^i)$ to mark itself in the ES-trees of each vertex at level $i$ that is at distance at most $2(1/\epsilon)^i$. If the level $i$ of $c$ is changed again, run a BFS to unmark $c$ at all ES-trees rooted at centers at level $i$ and at distance at most $2(1/\epsilon)^i$.

We use the algorithm presented in claim \ref{clm:maintainlevels} to maintain the levels. Since the marked vertices in $B(c, 2(1/\epsilon)^i)$ are at all times the vertices at the level $i$ that are at distance at most $2(1/\epsilon)^i$, inserting $c$ into $A_i$ once there are no marked vertices in its ES-tree enforces property \ref{prop:Aisalwaysfilled}.

Since we only remove a center $c$ from a set $A_i$ due to a level change, and since levels are monotonically decreasing over time, property \ref{prop:leaveA} is satisfied.

Finally, the running time required to maintain all ES-trees (including an ES-tree that monitors when vertices in $C_{\mu, \tau}$ transition to become $c \in C_{\mu, \tau}$ is $(\mu, 4\tau(1/\epsilon)^{k+1})$-\textit{light} and the level monitoring) is 
\[O(|C_{\mu, \tau}| \sum_{i \in [0, k]} \mu \tau(1/\epsilon)^i + |C_{\mu, \tau}| \mu (1/\epsilon)^{k+1}) = O(n \mu (1/\epsilon)^{k+1})
\]
since $\epsilon \leq 1/2$.
\end{proof}

Finally, we can state the emulator formally. 

\begin{definition}
Given the dynamic sets $A_0, A_1, .., A_k$ as described above, we define the weighted emulator $H^s_{\mu, \tau, \epsilon}$ such that:
\begin{itemize}
    \item each vertex $v \in V(H^s_{\mu, \tau, \epsilon})$ corresponds to a vertex in $V$ or a connected component in $G^*[\textsc{Near-Heavy}(\mu, 4\tau(1/\epsilon)^{k+1}, \tau)]$ (we call these vertices the \textit{component} vertices), and
    \item for every component vertex $c \in V(H^s_{\mu, \tau, \epsilon})$ corresponding to a connected component we have the edge $(c, v)$ for any $v \in V$ in $E(H^s_{\mu, \tau, \epsilon})$, and
    \item for every center $c \in A_i$, and $v \in \textsc{Cluster}_{\mu,\tau}(c)$ we have an edge $(c, v)$ to  $E(H^s_{\mu, \tau, \epsilon})$, and
    \item $E(H^s_{\mu, \tau, \epsilon})$ contains all edges $\{s\} \times B_{G^*}(s, 2\tau(1/\epsilon)^{k+1})$ that is to the vertices close to $s$.
\end{itemize}
We let the edge weights correspond to the distance between the endpoints in $G^*$ rounded up to the nearest multiple of $\epsilon\tau$, except for the edges incident to component vertices which are maintained with edge weight $\epsilon\tau$.
\end{definition}

The observant reader might have observed that we could have defined a conceptually easier emulator if we would have contracted vertices in the same connected components in $G^*[\textsc{Near-Heavy}(\mu, 4\tau(1/\epsilon)^{k+1}, \tau)]$ as we did in claim \ref{clm:guaranteeContractions}. However, the data structure that maintains the distance estimates using $H^s_{\mu, \tau, \epsilon}$ requires for technical reasons that each vertex in $V$ is explicitly present in the graph at all stages. We however observe that claim \ref{clm:guaranteeContractions} extends to give an upper bound for any underestimate on a shortest path in $H^s_{\mu, \tau, \epsilon}$.

\begin{lemma}
\label{lma:maintainEmu}
We can maintain the emulator $H^s_{\mu, \tau, \epsilon}$ in total time $O((n\mu+m\tau)(1/\epsilon)^{k+1} + m \log^2 n/ \log \log n)$.
\end{lemma}
\begin{proof}
On closer inspection, the edge weights of all edges incident to center vertices $c \in A_i$ can be maintained using the data structure in claim \ref{clm:maintainA}. Thus, maintaining the sets $A_i$ and these edge weights can be implemented in $O(n\mu (1/\epsilon)^{k+1})$ time. 

We can further use the data structure introduced in claim \ref{clm:heavysubgraphmaintain}, to detect when we have to insert a new vertex to $V(H^s_{\mu, \tau, \epsilon})$, i.e. every time the data structure returns a set $C_1$ that resulted from splitting a connected component $C$ previously represented by vertex $c$ into $C_1$ and $C_2 = C \setminus C_1$ (where $C_1$ is smaller than $C_2$) we add a new vertex $c_1$, remove the edges $\{c\} \times C_1$ from the emulator and add the new edges $\{c_2\} \times C_1$. Then $c$ represents the component $C_2$. This implies that every vertex changes its corresponding component vertex at most $O(\log n)$ times (because every time the component it is contained in halves). Since operations take linear time in the number of edge changes, we can bound the total update time due to component vertices and edges by $O(n \log n)$ which is subsumed in the update time to monitor the near-heavy/near-lightness of vertices.

Finally, we run an ES-tree from $s$ to depth $2\tau(1/\epsilon)^{k+1}$ on $G^*$ in time $O(m \tau(1/\epsilon)^{k+1})$ time to maintain all weights of the edges $\{s\} \times B_{G^*}(s, 2\tau(1/\epsilon)^{k+1})$ (and in fact monitors which edges can be removed since their endpoint leaves the ball). This completes the construction and analysis of the emulator.
\end{proof}

\section{Maintaining Distances using the Emulator}
\label{sec:usingEmulator}

Finally, we are ready to describe the data structure we use to find approximate shortest $s$ to $v$ distances using $H^s_{\mu, \tau, \epsilon}$. Whilst it is useful to have the monotone ES-tree as a conceptual starting point, our data structure requires more involved operations that are tailored towards maintaining $H^s_{\mu, \tau, \epsilon}$ efficiently.

To simplify presentation, we assume henceforth that $1/\epsilon$ and $\frac{m}{\mu}$ are integer. Further, we focus in this section on a specific scale of distances that we want to maintain, say we want to find estimates for the distances $\mathbf{dist}_{G^*}(u,v) \in [2^j, 2^{j+1})$ for some integer $\lg \frac{m}{\epsilon^{k+5} \mu} \leq j < \lg n$. We then use lemma \ref{lma:maintainEmu} to obtain an emulator $H^s_{\mu, \tau_j, \epsilon}$ where we set $\tau_j = \frac{\mu 2^j \epsilon^{k+4}}{m}$ (our restriction on $j$ implies that $\epsilon\tau_j$ is an integer). Further, we use the sets $A_0, A_1, .. , A_k$ used to maintain the emulator explicitly. Let us now state a definition of our data structure.

\begin{definition}
\label{def:almostMES}
Let an \textit{almost}-MES-tree $\mathcal{E}_j$ be a data structure that given a source $s \in V$, a depth $d > 0$, the sets $A_0, A_1, .. , A_k$ and $H^s_{\mu, \tau_j, \epsilon}$ maintains a distance estimate $\hat{l}(v)$ for every vertex $v \in V$ as follows:
\begin{itemize}
    \item Each distance estimate $\hat{l}(v)$ is initialized to the distance from $s$ to $v$ in the initial graph $H^s_{\mu, \tau_j, \epsilon}$. 
    \item If a distance estimate $\hat{l}(v)$ is at any stage increased beyond $d$, the vertex $v$ is removed from $\mathcal{E}_r$.
    \item If a vertex $v$ transitions from being $(\mu, 4\tau_j(1/\epsilon)^{k+1}, \tau_j)$-\textit{near-heavy} to $(\mu, 4\tau_j(1/\epsilon)^{k+1}, \tau_j)$-\textit{near-light}, we decrease $\hat{l}(v)$ by $2\epsilon\tau_j$.
    \item If a center $c$ is added to some $A_i$, $c$ checks each vertex $v \in \textsc{Core}_{\mu, \tau_j}(c)$ and updates its own distance estimate $\hat{l}(c) = \max_{v \in \textsc{Core}_{\mu, \tau_j}(c)} \hat{l}^{OLD}(v) - 8 \tau_j(1/\epsilon)^{i}$ where $\hat{l}^{OLD}(v)$ denotes the distance estimate obtained at the last stage for $v$. (This operation might decrease the distance estimate foiling the monotonicity of the function $\hat{l}$).
    \item Edge updates are handled in batches that correspond to the changes in $H^s_{\mu, \tau_j, \epsilon}$ during one stage. After all level updates due to the previous two operations are completed and edge changes are executed, we \uline{consolidate} all distance estimates according to the rule,
    \[
        \hat{l}^{NEW}(v) = \max\{ \hat{l}(v), \min_{(u,v) \in E(H^s_{\mu, \tau_j, \epsilon})} \hat{l}^{NEW}(u) + w(u,v)\}.
    \]
    Observe that this rule is the same rule as in a regular ES-tree procedure.
    \item After the new estimates are found, for any center $c \in A_i$ with $\hat{l}^{NEW}(v) > \hat{l}(v)$, we update for each vertex $v \in \textsc{Core}_{\mu, \tau_j}(c)$ the distance estimate $\hat{l}^{NEW}(v)$ to value 
    \[
    \max\{\hat{l}^{NEW}(v), \hat{l}^{NEW}(c) - \tau_j(1/\epsilon)^{k+1}\}.
    \]
    We say that $c$ \uline{drags} along the vertices in $\textsc{Core}_{\mu, \tau_j}(c)$. Whilst these updates can be evaluated in arbitrary order, dragging a center $c' \in A_j$ resulting in an estimate increase triggers a re-evaluation of the same procedure on $c'$. We do however not re-evaluate the \uline{consolidation} procedure. Observe that this might trigger a series of updates at the next stage. 
\end{itemize}
\end{definition}

We point out that the description above uniquely defines the values of distance estimates after each stage. Let us first establish that the data structure indeed maintains distance estimates that reflect actual distances in $G^*$. We then establish in the next section that the \textit{almost}-MES-tree can be implemented efficiently by designing the \uline{consolidation} procedure carefully.

Let us for convenience define $d' = (1-2\epsilon)d - 3\epsilon\tau_j - 4\tau_j(1/\epsilon)^{k+1}$. 

\begin{lemma}
\label{lma:stretchMES}
The \textit{almost}-MES-tree as defined above maintains for every vertex $v \in V$, with $\mathbf{dist}_{G^*}(s,v) \leq d'$ a distance estimates $\hat{l}(v)$ such that 
\[
\mathbf{dist}_{G^*}(s,v) - \epsilon2^j \leq \hat{l}(v) \leq (1+\epsilon)\mathbf{dist}_{G^*}(s,v) + 3\epsilon\tau_j.
\]
\end{lemma}
\begin{proof}
In the overview section, we already sketched a proof for the upper bound of the stretch on any version of $H^s_{\mu, \tau_j, \epsilon}$. Here, we use a similar approach to prove the upper bound on $\hat{l}(v)$ by using induction.

\begin{claim}
For any $\epsilon < 1/10$, $\tau \geq 1$, at any stage $i \geq 0$, for any $0 \leq \delta \leq d'$, 
\begin{enumerate}
    \item we have for every vertex $v \in V$ that is incident to a $(\mu, 4\tau_j(1/\epsilon)^{k+1}, \tau_j)$-\textit{near-light} vertex and is at distance $\delta$ from $s$ in $G^*$ that $\hat{l}(v) \leq (1+\epsilon)\mathbf{dist}_{G^*}(s,v) + \epsilon\tau_j$, and
    \item we have for every vertex $(\mu, 4\tau_j(1/\epsilon)^{k+1}, \tau_j)$-\textit{near-heavy} vertex $v \in V$, at distance $\delta$ from $s$ in $G^*$ that $\hat{l}(v) \leq (1+\epsilon)\mathbf{dist}_{G^*}(s,v) + 3\epsilon\tau_j$, and
    \item for every active center $c \in A_i$ at distance at most $\delta + 4\tau_j(1/\epsilon)^{i}$ from $s$ in $G^*$, we have that $\hat{l}(v) \leq \max\{(1+\epsilon)\mathbf{dist}_{G^*}(s,c) - 6\tau_j(1/\epsilon)^{i}, \mathbf{dist}_{G^*}(s,c) + \epsilon\tau_j\}$.
\end{enumerate}
\end{claim}
\begin{proof}
For the proof of this claim, we use the convention of letting $\hat{\pi}_{u,v}$ refer to the shortest path from $u$ to $v$ in $H^s_{\mu, \tau_j, \epsilon}$. We prove the claim using outer induction on the stages and inner induction on $\delta$.\\

\noindent\uline{Base Case (Outer induction on stages): Stage $0$.} 
\begin{itemize}
    \item \uline{Base Case (Inner induction on $\delta$): $\delta \leq \tau_j(1/\epsilon)^{k+1}$.} By definition of $H^s_{\mu, \tau_j, \epsilon}$, there is a direct edge from $s$ to $v$ (observe that this is true for all $v$ in $B_{G^*}(s, 2\tau_j(1/\epsilon)^{k+1})$, thus this is also true for all active centers). Since this edge has weight at most $\lceil \mathbf{dist}_{G^*}(s,v) / \epsilon\tau_j\rceil \epsilon\tau_j \leq \mathbf{dist}_{G^*}(s,v) + \epsilon\tau_j$ both properties are satisfied. 
    
    \item \uline{Inductive step (Induction on $\delta$): $\delta > \tau_j(1/\epsilon)^{k+1}$.} Let us prove in fact a slightly stronger statement: we can find for each vertex $v$ under consideration at the current $\delta$ a single edge $(u,v)$ such that $u$ is significantly closer to $s$ and we can therefore use the inner IH on $u$ to establish the final claim (except for heavy vertices which require two edges). 
    
    Let us start by observing that for every $v \in V$ where $v$ is incident to a $(\mu, 4\tau_j(1/\epsilon)^{k+1}, \tau_j)$-\textit{near-light} vertex, it is $(\mu, 4\tau_j(1/\epsilon)^{k+1} - 2\tau_j - 1)$-\textit{light}. Therefore there exists a center $c = p_{C_{\mu, \tau_j}}(v)$ in $C_{\mu, \tau_j}$ at distance at most $\tau_j$ (by lemma \ref{lma:coverefficientSimple}). It is straight-forward to establish that $c$ is $(\mu,  2\tau_j(1/\epsilon)^{k+1})$-\textit{light} and therefore  $i = \textsc{ClusterLevel}_{\mu, \tau}(c)$ is a value in $[0,k]$ by claim \ref{clm:maintainlevels}. By claim \ref{clm:maintainA} either $c$ is active or there exist another vertex $c'$ in $A_i$ that is at distance at most $2\tau_j(1/\epsilon)^i$ from $c$. Let $a$ denote $c$ in the former case, or $c'$ in the latter. 
    
    We have $\mathbf{dist}_{G^*}(v, a) \leq 2\tau_j(1/\epsilon)^i + 2\tau_j$. We have therefore that $\mathbf{dist}_{G^*}(s,a) < \delta + 4\tau_j(1/\epsilon)^i$. Thus, we can invoke the inner IH on $a$. Together, with the observation that the upper bound on $\mathbf{dist}_{G^*}(v, a)$ also implies that there is an edge $(a,v)$ in $H^s_{\mu, \tau_j, \epsilon}$, we obtain that
    \[
    \hat{l}(v) \leq \hat{l}(a) + w(a,v) \leq (1+\epsilon)\mathbf{dist}_{G^*}(s,v) - 6\tau_j(1/\epsilon)^i + 2\tau_j + 3\epsilon \tau_j \leq (1+\epsilon)\mathbf{dist}_{G^*}(s,v)
    \]
    where we again invoke the triangle inequality, establishing the first property of the claim.
    
    If $v$ is not incident to a $(\mu, 4\tau_j(1/\epsilon)^{k+1}, \tau_j)$-\textit{near-light} vertex, it is in a component in $G^*[\textsc{Near-Heavy}(\mu, 4\tau_j(1/\epsilon)^{k+1}, \tau_j)]$. Let $c$ be the vertex that corresponds to this component in $H^s_{\mu, 4\tau_j(1/\epsilon)^{k+1}, \epsilon}$. Further let $w$ be the vertex in the connected component corresponding to $c$ that is closest to $s$ on the path $\pi_{s,v}$. Clearly, $w$ is closer to $s$ than $v$ and incident to a $(\mu, 4\tau_j(1/\epsilon)^{k+1}, \tau_j)$-\textit{near-light} vertex. Thus, we can invoke the inner IH on $w$ and obtain that $\hat{l}(w) \leq (1+\epsilon) \mathbf{dist}_{G^*}(s,w)+\epsilon\tau_j$. But then using the edges $(w,c)$ and $(c,v)$, we must have a path from $s$ to $v$ of weight $(1+\epsilon) \mathbf{dist}_{G^*}(s,w)+ \epsilon\tau_j + 2\epsilon\tau_j < (1+\epsilon) \mathbf{dist}_{G^*}(s,v)+3\epsilon\tau_j$, as required.
    
    Finally, let us establish the upper bound on the distance estimates for any $c \in A_i$. We therefore let $w$ denote the vertex closest to $s$ on $\pi_{s,c} \cap \textsc{Cluster}_{\mu, \tau_j}(c)$. By choice of $w$, we have $\mathbf{dist}_{G^*}(s,w) = \mathbf{dist}_{G^*}(s,c) - \tau_j(1/\epsilon)^{i+1} = \delta + 4\tau_j(1/\epsilon)^{i} - \tau_j(1/\epsilon)^{i+1} < \delta$. Thus, we can invoke the inner IH on $w$ to obtain that $\hat{l}(w) \leq (1+\epsilon)\mathbf{dist}_{G^*}(s,w)$. Using the edge $(w,c)$ in conjunction to this upper bound on $\hat{l}(w)$, we can also upper bound $\hat{l}(c)$ by $(1+\epsilon)\mathbf{dist}_{G^*}(s,v) + \tau_j(1/\epsilon)^{i+1} \leq (1+\epsilon)\mathbf{dist}_{G^*}(s,w) - 4\epsilon\tau_j(1/\epsilon)^{i+1} \leq (1+\epsilon)\mathbf{dist}_{G^*}(s,w) - 4\tau_j(1/\epsilon)^{i}$.
\end{itemize}

Finally, we did not discuss the effect of setting distance estimates to $\infty$ after they were increased beyond $d$. However, by choice of $d'$, it is straight-forward to check the proof above to see that none of our proofs relies on a distance estimate that is a distance larger $d$ (in fact, this is a rather technical detail arising from our data structure). This will also be the case for the inductive step that is proved next and we therefore omit the discussion.\\

\noindent\uline{Induction (Outer induction on stages): Stage $i$ where we assume that the assumption holds for all previous stages.} Before we again delve into a case discussion, let us first prove that the claim on the distance estimates claim is true just before the \uline{consolidation} procedure is invoked. 
    
By the outer IH, for each vertex $v$ at distance $\delta$, the distance estimate $\hat{l}(v)$ satisfied the guarantees given in the claim if it did not transition from being $(\mu, 4\tau_j(1/\epsilon)^{k+1}, \tau_j)$-\textit{near-heavy} to  $(\mu, 4\tau_j(1/\epsilon)^{k+1}, \tau_j)$-\textit{near-light} or joined a set $A_i$ for any $i$. 
    
Let us first consider the former case: if a vertex $v$ transitioned during the stage and didn't join a set $A_i$, then we still have by outer IH that $\hat{l}(v) \leq \mathbf{dist}_{G^*}(s,v) + 3\epsilon\tau_j$ (it actually states this property for the graph $G^*$ at stage $i-1$, however distances increase in decremental graphs over the stages and therefore we can consider the current version of the graph). But our estimate update procedure decreases the level of every transitioning vertex by $2\epsilon\tau_j$, thus the level satisfies the claim above.
    
Next, let us consider that $v$ joined $A_i$: Then we set $\hat{l}(c) = \max_{v \in \textsc{Core}_{\mu, \tau_j}(c)} \hat{l}^{OLD}(v) - 8 \tau_j(1/\epsilon)^{i}$. But since each of the vertices in $\textsc{Core}_{\mu, \tau_j}(c)$ are at distance at most $\tau_j(1/\epsilon)^{i}$ from $c$, we have by the outer IH on the vertices in $\textsc{Core}_{\mu, \tau_j}(c)$, that this sets $\hat{l}(v)$ to at most $(1+\epsilon)\mathbf{dist}_{G^*}(s,v) - 6 \tau_j(1/\epsilon)^{i}$.
    
Let us next observe that the proof given for stage $0$, extends to establish that after the \uline{consolidation} procedure has finished, all distance estimates still stipulate the claim since we always only expose the last edge on the path and use the claim to establish that the endpoint is significantly closer to $s$ or in the case of heavy vertices, use the edge $(w,c)$ and $(c,v)$ where $c$ is the component vertex, but since the component vertex is not involved in any cluster, its level remains always equal the one of the vertex in the connected component with the lowest distance estimate plus an additive error of $\epsilon\tau_j$ resulting from the edge, thus the two edges exposed always offer a legitimate path. 

Finally, we claim that the last step of the algorithm does not increase any distance estimate beyond the claimed upper bounds. We prove by induction on the number $t$ of centers that completed the \uline{dragging} step. In the base case $t=0$, this is vacuously true by our preceding discussion. Let us now show that it also holds for any $t > 0$, where the core of center $c \in A_i$ was updated, to prove the inductive step. By the IH, no distance estimate was increased beyond the claimed upper bound, before the \uline{dragging} step took place. But since this implies that $\hat{l}^{NEW}(c) \leq \mathbf{dist}_{G}(s,c)$ and every vertex $v$ in the core $\textsc{Core}_{\mu,\tau_j}(c)$ is at distance at most $\tau_j(1/\epsilon)^k$. Thus, setting the distance estimate of $v$ to $\max\{\hat{l}^{NEW}(v), \hat{l}^{NEW}(c) - \tau_j(1/\epsilon)^{k+1}\}$ can not push any vertex beyond the claimed bounds (assuming $\epsilon \leq \frac{1}{10}$).

This completes the inductive step and thereby establishes the claim.
\end{proof}

It remains to prove a lower bound on $\hat{l}(v)$. We therefore simply observe that claim \ref{clm:guaranteeContractions} upper bounds the negative additive error by $\frac{40 \tau_j(1/\epsilon)^{k+2} m}{\mu} = \epsilon2^j$ (by choice of $\tau_j$). Since the \uline{consolidation} procedure ensures that on finishing for each distance estimate $\hat{l}(v)$ there exists a $s$ to $v$ path in $H^s_{\mu, \tau_j, \epsilon}$ that is of weight at least $\hat{l}(v)$ (one could recursively expose the last edge on the $s$ to $v$ path using the \uline{consolidation}-rule). Since we might only overestimate estimates based on paths this procedure ensures that the lower bound of claim \ref{clm:guaranteeContractions} extends to the distance estimates. Further, after the \uline{consolidation} procedure, distance estimates can only be increased by the algorithm.
\end{proof}

\section{Implementing the almost-MES-tree efficiently}
\label{sec:implementMES}

Whilst a straight-forward implementation of the almost-MES-tree is possible, we need a rather small alternation of the classic implementation where edges are mapped to some endpoint and only distance estimate changes of that endpoint cause the algorithm to consider this edge whilst the other endpoint only finds the edge if it is an edge that satisfies the \uline{consolidation} rule. This alternation was previously introduced and described in  \cite{bernstein2017deterministic}. Let us now discuss a formal definition.

We let $E_{ALL}$ be the collection of all edges that appear in any version of the dynamic graph $H^s_{\mu, \tau_j, \epsilon}$. Here, we treat edges that were re-inserted as new edges, i.e. if $(u,v)$ appears in a version of the graph, is then removed and re-inserted, both versions appear in $E_{ALL}$. This ensures that there is a one-to-one correspondence between edges in $E_{ALL}$ and edge insertions into $H^s_{\mu, \tau_j, \epsilon}$.

\begin{definition}[c.f.  \cite{bernstein2017deterministic}, Def. 2.5]
A dynamic \emph{assignment} $A : E_{ALL} \rightarrow V(H^s_{\mu, \tau_j, \epsilon})$ is a function that maps each edge $(u,v)$ in $E_{ALL}$ to one of its endpoints $u$ or $v$. Further, $A$ must assign each edge $(u,v)$ to one of its endpoints at the moment when $(u,v)$ enters the edge set and cannot be changed afterwards.
\end{definition}

\begin{lemma}[ \cite{bernstein2017deterministic}, Lemma 2.2]
\label{lma:efficientalmostMES}
Consider an almost-MES-tree $\mathcal{E}_j$ as described in definition \ref{def:almostMES}. Then, given a dynamic assignment $A$, let $\Delta(u,v)$ be the total amount of distance estimate changes of $A(u,v)$ in $\mathcal{E}_j$ whilst $(u,v)$ is in the dynamic graph $H^s_{\mu, \tau_j, \epsilon}$. Then, we can implement the almost-MES-tree $\mathcal{E}_j$ as described in definition \ref{def:almostMES} with total time $O(|E_{ALL}| + \sum_{(u,v) \in E_{ALL}} \Delta(u,v) / \epsilon\tau_j)$ and constant worst-case query time.
\end{lemma}

Lemma \ref{lma:efficientalmostMES} is proved in  \cite{bernstein2017deterministic} (technically the statement of their lemma is slightly different but a careful inspection of their proof shows that it extends to our specific setting).

\begin{lemma}
\label{lma:totalrunningtime}
We can implement the almost-MES-tree $\mathcal{E}_j$ with total update time $\tilde{O}((n\mu + n^{1+ O(1/\sqrt{\log n})}d / \epsilon\tau_j) (1/\epsilon)^{O(\sqrt{\log n})})$.
\end{lemma}
\begin{proof}
Let us first construct a dynamic assignment $A$: we first recall that each edge $(u,v) \in E_{ALL}$ is either incident to an active center; to a component vertex or to $s$. In the first case, we assign the edge to the active center (or to an arbitrary endpoint if both endpoints are active centers). In the second case, we assign the edge to the endpoint in $V$. In the third case, we assign it to $s$. If an edge satisfies multiple cases, we handle it according to the first case that applies (i.e. if an edge has an endpoint $s$ and is also incident to an active center, we apply the assignment rule that we use if an active center is present). We use this assignment in lemma \ref{lma:efficientalmostMES}.

Let us first analyze the cost incurred by edges of one of the two latter edge types. Clearly, there are at most $O(n)$ edges in $E_{ALL}$ incident to $\{s\}$ (and not incident to an active center or a component vertex) since we only add each edge from $s$ to the initial graph and never re-insert it after removing it. Since $s$ always satisfies $\hat{l}(s) = 0$, the total cost incurred by these edges is thus $O(n)$. Further, observe that by construction of $H^s_{\mu, \tau_j, \epsilon}$ there are at most $O(n \log n)$ edges in $E_{ALL}$ incident to the component vertices. We further observe that these edges are incident to heavy vertices and the distance estimate $\hat{l}(v)$ of a vertex $v$ whilst being heavy is monotonically increasing. Since we no longer maintain $\hat{l}(v)$ after it reached distance $d$, the total cost incurred by edges of the second edge type is at most $O(n \log n + n \log n d /\epsilon\tau_j)$.

Finally, let us analyze the costs incurred by edges of the first edge type. Here, we point out that a careful inspection of definition \ref{def:almostMES} gives that any distance estimate $\hat{l}(v)$ is at most increased by an amount of $d + 1 + 2\epsilon\tau_j + \sum_{i \in [0,k]} 8\tau_j(1/\epsilon)^i \leq d + \tau_j(1/\epsilon)^{k+1}$ since $\hat{l}(v)$ can only be decreased due to becoming light or by entering some set $A_i$ (which it can only once during the course of the algorithm by claim \ref{clm:maintainA}) before it is removed once it has reached any value beyond $d$. 

Let us next consider the following accounting scheme: for every unit that $v$ has its distance estimate $\hat{l}(v)$ increased, it pays $n^{1/k}$ credits to every vertex $c \in A_i$ (for any $i$) with $v \in \textsc{Core}_{\mu, \tau_j}(c)$. It is straight-forward to bound the total amount of credits payed in this way by $kn^{1+ 1/k} (d + \tau_j(1/\epsilon)^{k+1})$ where we use that each vertex can only be in the core of one active center on any level by claim \ref{clm:maintainA}. 

With this in place, let us bound total cost induced by edges incident to active centers. We start by observing that for each active center $c \in A_i$, we have at most $\min\{\mu, n^{(i+1)/k}\}$ vertices in $\textsc{Cluster}_{\mu, \tau_j}(c)$ (by lightness and by definition of clusters). This immediately implies that we add $O(kn\mu)$ edges to $E_{ALL}$ of the first type. Now, let us focus on the stages $t_{enter}$ where $c$ enters $A_i$ and the stage $t_{left}$ where it leaves $A_i$ again. Let $\Delta_i(c)$ be the total amount of change in $\hat{l}(c)$ during these two stages. Let $\textsc{Core}_{\mu, \tau}^{FINAL}(c)$ be the core of $c$ at the stage $t_{left} - 1$. Then, $|\textsc{Core}_{\mu, \tau}^{FINAL}(c)| \geq n^{i/k}$ by claim \ref{clm:maintainlevels}. Now every vertex $v \in \textsc{Core}_{\mu, \tau}^{FINAL}(c)$ was initially at distance at most $(1/\epsilon)^{i}$ from $c$. Further, the \uline{dragging} step in the almost-MES-tree definition ensures that the vertices are at most at distance $\tau_j(1/\epsilon)^{k+1}$ at all stages in the interval $[t_{enter}, t_{left})$ from $c$. Thus, in particular, we have that at stage $t_{left} - 1$, they are still close to $c$. But then by the triangle inequality each vertex $v$ has increased its distance estimate during the interval $[t_{enter}, t_{left})$ by at least $\Delta_i(c) - (1/\epsilon)^{i} - \tau_j(1/\epsilon)^{k+1} < \Delta_i(c) - 2\tau_j(1/\epsilon)^{k+1}$. Thus, for all but $2\tau_j(1/\epsilon)^{k+1}$ increases of $\hat{l}(c)$ we can account for scanning the edges to each vertex in $\textsc{Cluster}_{\mu, \tau}(c)$ by using the credits payed by the vertices in $\textsc{Core}_{\mu, \tau}^{FINAL}(c)$ (remember that $\textsc{Core}_{\mu, \tau}^{FINAL}(c)$ is at most by factor $n^{1/k}$ smaller than $\textsc{Cluster}_{\mu, \tau}(c)$ at any stage). Since every vertex can only enter $A_i$ once, the total cost, not covered by the accounting scheme, amounts to at most $n 2\tau_j(1/\epsilon)^{k+1} \mu$. 
\end{proof}

\section{Putting it all together}
\label{sec:puttingTogether}

We let $d_{min} = 2\frac{m}{\epsilon^{k+5} \mu}$ and to obtain an distance estimate for a vertex $v \in V$ with $\mathbf{dist}(s,v) \leq d_{min}$ we maintain an ES-tree $\mathcal{E}_0$ from $s$ to depth $2\frac{m}{\epsilon^{k+5} \mu}$ that can return the answer. To handle larger distances, we maintain for each index $j \in [\lg d_{min}, \lg n]$ an almost-MES-tree $\mathcal{E}_j$ to depth $(1+\epsilon) 2^{j+1} + (1/\epsilon)^{k+2}$. Then, by lemma \ref{lma:stretchMES}, for every vertex $v \in V$ with $\mathbf{dist}_{G^*}(s,v) \in [2^j, 2^{j+1})$ there is a distance estimate $\hat{l}(v)$ in $\mathcal{E}_j$ that satisfies
\[
(1-\epsilon)\mathbf{dist}_{G^*}(s,v) \leq \hat{l}(v) \leq (1+2\epsilon)\mathbf{dist}_{G^*}(s,v).
\]
By dividing $\hat{l}(v)$ by $(1-\epsilon)$ and by re-scaling $\epsilon$ by a constant fraction, we obtain a $(1+\epsilon)$-approximate distance estimate. Finally, before we output this distance estimate, we check whether $s$ and $v$ are still in the same component of $G$, and if so we return the distance estimate and otherwise $\infty$.

Observe that by taking the distance estimate from the $\mathcal{E}_{j'}$ with the smallest $j'$ where $v$ is still in $\mathcal{E}_{j'}$, we are guaranteed to obtain a satisfying distance estimate (since it is guaranteed to be in $\mathcal{E}_j$ and since for any $j' \leq j$ the distance estimate improves).

Invoking lemma \ref{lma:maintainBalls} for $\mathcal{E}_0$ and lemma \ref{lma:totalrunningtime} for each $\mathcal{E}_j$, we obtain total running time,
\[
{O}\left( \left(\frac{m^2}{\mu} + n\mu + n^{1+ O(1/\sqrt{\log n})}d\right) (1/\epsilon)^{O(\sqrt{\log n})}\right).
\]
This term is optimized by choosing $\mu = m/\sqrt{n}$, where the term evaluates to 
\[
{O}(m n^{0.5 + O(1/\sqrt{\log n})} (1/\epsilon)^{O(\sqrt{\log n})}) = {O}(m n^{0.5 + o(1)}).
\]
Finally, we point out that the query time by finding the right index $j$ manually takes up to $O(\log n)$ time but by a standard technique (see for example  \cite{henzinger2016dynamic}) the distance estimates can be maintained in queues whose total running time is subsumed by our total update time bounds and which can be used to improve the query time to $O(1)$ in the worst-case. This establishes our main result, theorem \ref{thm:mainSSSP}.

\section{Conclusion}
\label{sec:conclusion}

In this article, we gave the first \emph{deterministic} data structure to solve the $(1+\epsilon)$-approximate single-source shortest path problem that achieves total update time $\tilde{O}(m n^{0.5 + o(1)})$ which is a polynomial improvement over all existing data structures for the problem for graph densities $m = O(n^{1.5 - o(1)})$. In doing so, we reached a natural barrier of all current approaches and presented various new techniques that might have applications beyond our algorithm.

Whilst our result manifest significant progress on the problem, it is still open whether the near-linear time algorithm for the problem in the randomized setting against an oblivious adversary can be matched \emph{deterministically} (or whether an algorithm can be developed to work against an adaptive adversary). If the running time can not be matched, it would be interesting to prove lower bounds although this requires extensive new lower bound techniques since all conditional lower bounds techniques for dynamic graph problems are designed against an oblivious adversary. 

Further, all sparse \emph{deterministic} algorithms for $(1+\epsilon)$-approximate SSSP currently have the drawback of not being able to return approximate shortest paths. Also our algorithm can not be extended to find shortest paths (at least not in a straight-forward manner) since we explicitly omit looking at dense subgraphs and therefore have no information about paths crossing them. However, this issue was recently resolved in dense graphs \cite{Chuzhoy:2019:NAD:3313276.3316320}.

Finally, it is also a major open question whether the total running time for $(1+\epsilon)$-approximate SSSP in \emph{directed} graphs can be improved and whether there exist deterministic algorithms that improve upon ES-trees. Recently, the same authors \cite{probstWulffNilsenDiSSSP} provided a significant improvement on the $O(mn^{0.9+o(1)})$ algorithm by Henzinger, Forster and Nanongkai \cite{henzinger2014sublinear, henzinger2015improved} that also works against an adaptive adversary, however, again, it cannot report paths. The simpler problem of maintaining single-source reachability was recently improved significantly by Chechik et al. \cite{chechik2016decremental} by a clever combination of techniques in  \cite{lkacki2011improved, roditty2016fully} and has since been further improved to near-optimality \cite{bernstein2019decremental}. However, again, it is wide open whether a deterministic algorithm beyond the ES-tree exists.

\paragraph{Acknowledgements}

The authors of the paper would like to thank Jacob Evald and Viktor Hansen for helpful comments on organization and correctness.


\printbibliography[heading=bibintoc] 

\pagebreak

\appendix

\section{Proof of Lemma \ref{lma:efficientCover}}
\label{sec:appendixCoverProof}

\begin{proof}
We again prove the claim for the values $\alpha = 2$ and $\gamma = 1$ but prove that if a vertex $v$ is $(\mu, 2\alpha^{j})$-\textit{light} then $p_j(v)$ is $(\mu, \alpha^{j})$-\textit{light}. 

We prove by induction. Since $C_0 = V$, the property is vacuously true for the base case since every vertex that is $(\mu, 2\alpha^{j})$-\textit{light} is $(\mu, \alpha^{j})$-\textit{light}. 

Let us take the induction step. For $j > 0$, we again use that every vertex that is $(\mu, 2^{j+1})$-\textit{light} is also $(\mu, 2^{j})$-\textit{light} and invoke the induction hypothesis to deduce that there exists a vertex $c' \in C_{j-1}$ with $\dist{v, c'} \leq 2^{j-1}$. We have that $B(c', 2^j) \subseteq B(v,2^{j+1})$ since every vertex $x \in B(c', 2^j)$ satisfies $\dist{x,v} \leq 2^{j-1} + 2^j \leq  2^{j+1}$. Thus, $c'$ satisfies the first property of $j$-eligibility. If the vertex $c'$ is not in $C_j$, then there is a vertex $c'' \in C_j$ at distance at most $2^j$ from $c'$ and we conclude that the distance from $c''$ to $v$ satisfies $\dist{v, c''} \leq \dist{v, c'} + \dist{c', c''} \leq 2^{j-1} + 2^{j-1} = 2^j$. It also follows that for every such $c'$ or $c''$, the vertex is $(\mu, 2^{j})$-\textit{light} since $v$ is $(\mu, 2^{j+1})$-\textit{light} and at distance at most $2^j$.

Let us now analyze the number of vertices that are in a set $C_{j}$. Any two vertices $u,v \in C_{j}$ are at distance $\dist[\mathcal{G}^*]{u,v} \geq 2^{j-1}$, thus the ball $B(u, 2^{j-1}/2)$ is disjoint from all balls $B(v, 2^{j-1}/2)$ for every $u,v \in V, v \neq u$. Since $\mathcal{G}^*$ is connected by definition, we have that $|B(v, 2^{j-1}/2)| \geq 2^{j-1}/2$ and by simple counting arguments, there can be at most $2n/2^{j-1} = 4n / 2^j$ such disjoint balls and therefore vertices in $C_j$. 

Let us finally bound the running time of the algorithm. For each level $j \in [0, \lfloor\lg(n)\rfloor]$, we run by the preceding analysis an ES-tree as described in \ref{lma:maintainBalls} from $O(n/2^j)$ vertices to depth $r = O(2^j)$. By lemma \ref{lma:maintainBalls}, we can detect whether a vertex is $(\mu, r)$-\textit{light} in $O(\mu r)$ time. We can thus monitor the first property of $j$-eligibility for a vertex in $O(\mu 2^j)$ time. Since we maintain ES-trees to monitor this property, we can also use them to keep count of incident vertices that are in $C_{j}$ without increasing the running time. To maintain $p_j(v) \in C_j$ for every vertex that is $(\mu, 2^{j+1})$-\textit{light}, we observe that by definition \ref{def:cover} and our preceding discussion, there is a vertex $c$ in $C_j$ at distance at most $2^j$ that is $(\mu, 2^j)$-\textit{light}. Thus, we can explicitly maintain the distance $\mathbf{dist}_{\mathcal{G}}(v,c)$ by using the ES-tree of $c$. We store at each $v$ a list for each such center where it is contained in the ES-tree and maintain a increasing order by distance to $v$. If the distance to a center increases by one then the ES-tree spends at least constant time and we use an additional constant time to check whether another vertex in the ordered list is now closer to $v$ and if so assign it to $p_j(v)$. We thus conclude that the overall time can be bound by $O(\sum_{j = 0}^{\lg n} \mu 2^{j+1} n/2^j) = O(\mu n  \log{n})$.

To maintain a layered $(\mu, \alpha, \gamma)$-cover $\mathcal{C}'= \{C'_0, C'_1, .., C'_{\lfloor\log_{\alpha}(n/\gamma)\rfloor}\}$ for any $\alpha \geq 2$ and $\gamma \geq 1$, we can again simply maintain a layered $(\mu, 2, 1)$-cover $\mathcal{C}= \{C_0, C_1, .., C_{\lfloor\log_{\alpha}(n)\rfloor}\}$ and map each index $j \in [0,\lfloor\log_{\alpha}(n/\gamma)\rfloor]$ to $j' = \lfloor j * \log_{2}(\alpha) \rfloor$ and return the cover $C_{j'}$. 
\end{proof}

\end{document}